\documentclass[10pt]{article}
\usepackage{amsfonts,color}
\usepackage{amssymb,amsmath,amsthm,latexsym}
\usepackage{enumerate}
\usepackage{amsmath,amsthm,amssymb,arydshln}
\allowdisplaybreaks[4]
\newtheorem{Theorem}{Theorem}[section]
\newtheorem{lem}[Theorem]{Lemma}
\newtheorem{Remark}[Theorem]{Remark}

\newtheorem{Corollary}[Theorem]{Corollary}

\newtheorem{Example}[Theorem]{Example}
\numberwithin{equation}{section}

\usepackage{enumerate} 

\begin{document}
	
	\title{New upper bounds on the number of non-zero weights of constacyclic codes}
	\author{Li Chen, Yuqing Fu and Hongwei Liu\footnote{E-mail addresses:
			chenxiaoli@mails.ccnu.edu.cn (L. Chen),
			yuqingfu@mails.ccnu.edu.cn (Y. Fu),
			hwliu@ccnu.edu.cn (H. Liu)}}
	\date{\small
		School of Mathematics and Statistics, Central China Normal University, Wuhan, 430079, China\\
	}
	\maketitle
	{\noindent\small{\bf Abstract:} For any simple-root constacyclic code $\mathcal{C}$ over a finite field $\mathbb{F}_q$, as far as we know, the group $\mathcal{G}$ generated by the multiplier, the constacyclic shift and the scalar multiplications is the largest subgroup of the automorphism group ${\rm Aut}(\mathcal{C})$ of $\mathcal{C}$. In this paper, by calculating the number of $\mathcal{G}$-orbits of $\mathcal{C}\backslash\{\bf 0\}$, we give an explicit upper bound on the number of non-zero weights of $\mathcal{C}$ and present a necessary and sufficient condition for $\mathcal{C}$ to meet the upper bound. Some examples in this paper show that our upper bound is tight and better than the upper bounds in [Zhang and Cao, FFA, 2024]. In particular, our main results provide a new method to construct few-weight constacyclic codes. Furthermore, for the constacyclic code $\mathcal{C}$ belonging to two special types, we obtain a smaller upper bound on the number of non-zero weights of $\mathcal{C}$ by substituting $\mathcal{G}$ with a larger subgroup of ${\rm Aut}(\mathcal{C})$. The results derived in this paper generalize the main results in [Chen, Fu and Liu, IEEE-TIT, 2024]}.
	
	\vspace{1ex}
	
	
	\section{Introduction}
	Throughout this paper, let $\mathbb{F}_{q}$ denote the finite field with $q$ elements, where $q$ is a prime power.
	Let $n$ be a positive integer that is coprime to $q$.
	An $[n,k,d]$ linear code $\mathcal{C}$ over $\mathbb{F}_{q}$ is defined as a $k$-dimensional subspace of $\mathbb{F}_{q}^{n}$ with minimum (Hamming) distance $d$.
	Let $A_i$ be the number of codewords of (Hamming) weight $i$ in $\mathcal{C}$.
	The polynomial $1 + A_1 x + \cdots + A_n x^n$ is called the weight enumerator of $\mathcal{C}$
	and $(1,A_1,\cdots,A_n)$ is called the weight distribution (or weight spectrum) of $\mathcal{C}$.
	The weight distribution contains important information for estimating the probability and capability of error correction of a code.
	Therefore, the weight distribution attracts much attention in coding theory, and determining the weight distribution of linear codes has also become a hot topic.
	For a linear code $\mathcal{C}$, let $t$ be the number of nonzero $A_i$'s in the weight distribution.
	Then the code $\mathcal{C}$ is called a $t$-weight code.
	Linear codes with few weights are important in secret sharing \cite{25}, \cite{26}, authentication codes \cite{27}, association schemes \cite{28} and strongly regular graphs \cite{29}.
	
	Let $\mathbb{F}^{*}_q$ denote the multiplicative group of $\mathbb{F}_{q}$.
	For $\lambda \in \mathbb{F}^{*}_q$, a linear code $\mathcal{C}$ is called a $\lambda$-constacyclic code if $(\lambda a_{n-1},a_0,a_1,\cdots, a_{n-2})\in \mathcal{C}$ for every $\mathbf{c}=(a_0, a_1,\cdots,a_{n-1})\in \mathcal{C}$.
	It is well known that a $\lambda$-constacyclic code of length $n$ over $\mathbb{F}_{q}$ can be identified as an ideal of the quotient ring $\mathbb{F}_{q}[x]/\langle x^{n}-\lambda \rangle$ via the $\mathbb{F}_q$-module isomorphism $\pi : \mathbb{F}_{q}^{n} \rightarrow \mathcal{R}^{(q)}_{n,\lambda}:=\mathbb{F}_{q}[x]/\langle x^{n}-\lambda \rangle$ given by
	\vspace{-1.5ex}
	$$(a_0, a_1,\cdots,a_{n-1}) \mapsto a_{0}+a_{1}x+\cdots+a_{n-1}x^{n-1} \pmod{x^n - \lambda},\vspace{-0.75ex}$$
	where $\langle x^{n}-\lambda \rangle$ is an ideal of the polynomial ring $\mathbb{F}_{q}[x]$ generated by $x^{n}-\lambda$.
	If $\lambda = 1$, $\lambda$-constacyclic codes are just cyclic codes;
	and if $\lambda = -1$, $\lambda$-constacyclic codes are known as negacyclic codes.
	As we all know, a linear code $\mathcal{C}$ of length $n$ over $\mathbb{F}_{q}$ corresponds to an $\mathbb{F}_{q}$-subspace of the algebra $\mathcal{R}^{(q)}_{n,\lambda}$.
	Moreover, $\mathcal{C}$ is $\lambda$-constacyclic if and only if the corresponding subspace is an ideal of $\mathcal{R}^{(q)}_{n,\lambda}$.
	A $\lambda$-constacyclic code $\mathcal{C}$ is called irreducible if $\mathcal{C}$ is a minimal ideal of $\mathcal{R}^{(q)}_{n,\lambda}$.
	When $n$ is coprime to the characteristic of $\mathbb{F}_q$, a constacyclic code of length $n$ over $\mathbb{F}_q$ is called a simple-root constacyclic code; otherwise it is called a repeated-root constacyclic code.
	Constacyclic codes form an algebraically rich family of error-correcting codes, and are generalizations of cyclic and negacyclic codes.
	These codes can be efficiently encoded using shift registers and can be easily decoded due to their rich algebraic structure, which explain their preferred role in engineering.
	
	The number of non-zero weights of a linear code plays a crucial role in the theory of error-correcting codes, and its research topic has always attracted people's interest.
	In 1969, Assmus and Mattson \cite{3} derived a relationship between codes and designs in terms of the number of non-zero weights of a linear code.
	In 1973, Delsarte studied the number of distinct distances for a code $\mathcal{C}$. In the linear case,
	this reduces to studying the number of distinct weights of the given code \cite{8},  which is consistent with studying the number of non-zero weights.
	In that work, the author emphasizes the importance of this parameter by analyzing its relationships with the number of non-zero weights of the dual code, as well as the minimum distance of both the code and its dual.
	By analyzing these parameters, the author derived many interesting results about the properties of distance.
	In particular, these parameters are used to calculate the coset weight distributions of a code.
	In addition, the number of non-zero weights of a code has close connection with orthogonal arrays and combinatorial designs (see \cite{8}, \cite{9}).
	
	For a general linear code, it seems very difficult to obtain an explicit formula for the number of non-zero weights of the code. A more modest objective is to establish acceptable bounds on the number of non-zero weights of a linear code.
	Indeed, several recent works have looked into the upper and lower bounds on the number of non-zero weights of a linear code.
	Alderson \cite{1} determined necessary and sufficient conditions for the existence of linear full weight spectrum codes over $\mathbb{F}_{q}$, i.e., linear codes satisfying that there exist codewords of each weight less than or equal to the code length.
	Ding and Yang \cite{Ding13}  studied the weight distributions of irreducible cyclic codes, and established a lower and upper bound on the number of non-zero weights of these codes.
	Shi et al. \cite{19} conjectured that for a linear code of dimension $k$ over $\mathbb{F}_{q}$, the largest number of non-zero weights of this code is bounded from above by $(q^{k}-1)/(q-1)$, and proved that the bound is sharp for binary codes and for all $q$-ary codes of dimension $k = 2$.
	The conjecture was completely proved by Alderson and Neri \cite{2}.
	Shi et al. \cite{17} presented lower and upper bounds on the largest number of non-zero weights of cyclic codes, and gave sharper upper bounds for strongly cyclic codes, where the periods of their non-zero codewords are equal to the code length.
	Shi et al. \cite{18} investigated the largest number of non-zero weights of quasi-cyclic codes, and presented several lower and upper bounds on the largest number of non-zero weights of quasi-cyclic codes.

	Chen and Zhang \cite{7} observed that the number of non-zero weights of a linear code is bounded from above by the number of orbits of the automorphism group (or a subgroup of the automorphism group) acting on the code, with equality if and only if any two codewords with the same weight belong to the same orbit.
	Let $\mathcal{C}$ be a simple-root cyclic code of length $n$ over $\mathbb{F}_{q}$
	and let $\mathcal{G}$ be the subgroup of the automorphism group ${\rm Aut}(\mathcal{C})$ of $\mathcal{C}$ generated by the cyclic shift and the scalar multiplications.
	The authors derived an explicit upper bound on the number of non-zero weights of $\mathcal{C}$ by calculating the number of $\mathcal{G}$-orbits of $\mathcal{C}^{*}=\mathcal{C}\backslash\{0\}$, and established a necessary and sufficient condition for codes meeting the bound.
	Li and Shi \cite{20} established a tight upper bound on the number of non-zero weights of a simple-root quasi-cyclic code. Zhang and Cao \cite{21} established a tight upper bound on the number of non-zero weights of a simple-root constacyclic code. In \cite{22}, Chen, Fu and Liu improved the upper bound on the number of non-zero weights of a simple-root cyclic code $\mathcal{C}$ in \cite{7} by replacing $\mathcal{G}$ with larger subgroups of ${\rm Aut}(\mathcal{C})$.
	
	Motivated by the work \cite{22}, the objective of this paper is to establish a smaller upper bound than those in \cite{21} on the number of non-zero weights of a simple-root $\lambda$-constacyclic code $\mathcal{C}$.
	Using the observation of Chen and Zhang in \cite{7} and Burnside's lemma, the problem is transformed into finding larger subgroups of ${\rm Aut}(\mathcal{C})$ than that in \cite{21}.
	It is well known that both the $\lambda$-constacyclic shift and the scalar multiplications
	are automorphisms of $\mathcal{C}$.
	In addition, we note that the multiplier $\mu_{q}$ defined on $\mathcal{R}^{(q)}_{n,\lambda}$ by
	\vspace{-1ex} $$\mu_{q}\big(\sum_{i=0}^{n-1}a_{i}x^{i}\big)=\sum_{i=0}^{n-1} a_{i}x^{qi}~({\rm mod}~x^n-\lambda)$$
	is also an automorphism of $\mathcal{C}$. Let $\mathcal{G}'$ be the subgroup of ${\rm Aut}(\mathcal{C})$ generated by the $\lambda$-constacyclic shift and the scalar multiplications and $\mathcal{G}''$ be the subgroup of ${\rm Aut}(\mathcal{C})$ generated by the multiplier $\mu_{q}$, the $\lambda$-constacyclic shift and the scalar multiplications.
	To the best of our knowledge, $\mathcal{G}''$ is the largest subgroup of ${\rm Aut}(\mathcal{C})$ for any simple-root $\lambda$-constacyclic code $\mathcal{C}$ over $\mathbb{F}_{q}$. Clearly, $\mathcal{G}'$ is a subgroup of $\mathcal{G}''$, and the number of $\mathcal{G}''$-orbits of $\mathcal{C}^{*}$ is less than or equal to the number of $\mathcal{G}'$-orbits of $\mathcal{C}^{*}$.
	Therefore, we need to find the number of $\mathcal{G}''$-orbits of $\mathcal{C}^{*}$, which then naturally derive smaller upper bounds on the number of non-zero weights of $\mathcal{C}$.
	However, intuitively, the structure of $\mathcal{G}''$ is more complicate than that of $\mathcal{G}'$, which implies that  finding the number of $\mathcal{G}''$-orbits of $\mathcal{C}^{*}$ may be more difficult than finding the number of $\mathcal{G}'$-orbits of $\mathcal{C}^{*}$.
	Comparing the proofs in our paper and those in \cite{21}, we need more subtle calculations and have to get around new difficulties that have not arisen before.
	
	In this paper, by calculating the number of $\mathcal{G''}$-orbits of $\mathcal{C}^{*}$, we establish an explicit upper bound on the number of non-zero weights of $\mathcal{C}$ and present a necessary and sufficient condition for $\mathcal{C}$ to meet the upper bound. Many examples are presented to show that our upper bound is tight and strictly less than the upper bounds in \cite{21}.
	Moreover, comparing our results with those in \cite[subsection 3.2]{21}, our results remove the constrain that ${\rm gcd}\big( \frac{q-1}{r}, r \big) = 1$, i.e., our results hold for arbitrary simple-root $\lambda$-constacyclic code, where $r$ is the order of $\lambda$ in $\mathbb{F}_{q}^{*}$.
	In addition, for two special classes of constacyclic codes, we replace $\mathcal{G}''$ with lager subgroups of the automorphism groups of these codes, and then we obtain smaller upper bounds on the number of non-zero weights of these codes. The results derived in this paper improve and generalize some of the results in \cite{22} and \cite{21}.
	In particular, our main results provide a new method to construct few-weight constacyclic codes.
	
	This paper is organized as follows. In Section 2, we review some definitions and basic results about group action, $\lambda$-constacyclic codes and subgroups of the automorphism group of $\mathcal{R}^{(q)}_{n,\lambda}$. In Section 3, we propose the main results of this paper. This is divided into three subsections: in subsections 3.1 and 3.2, we present improved upper bounds on the number of non-zero weights of irreducible $\lambda$-constacyclic codes and general $\lambda$-constacyclic codes, respectively, by calculating the number of $\mathcal{G}''$-orbits of $\mathcal{C}\backslash \{\bf 0\}$. In subsection 3.3, for two special classes of constacyclic codes, we derive smaller upper bounds on the number of non-zero weights of these codes by replacing $\mathcal{G}''$ with lager subgroups of the automorphism groups of these codes. In Section 4, we conclude this paper with remarks and some possible future works.
	
	\section{Preliminaries}
	Throughout this paper,
	let $\mathbb{F}_{q}$ denote the finite field with $q$ elements  and let  $n>1$ be a positive integer relatively prime to $q$, where $q$ is a power of a prime.
	By $\mathbb{F}_{q}^{*}$, we denote the multiplicative group of $\mathbb{F}_{q}$.
	For $a\in \mathbb{F}_{q}^{*}$, ${\rm ord}(a)$ denotes the order of $a$ in $\mathbb{F}_{q}^{*}$.
	Let $\mathbb{Z}_{n}$ denote the residue class ring of the integer ring $\mathbb{Z}$ modulo $n$ and let $\mathbb{Z}_{n}^{*}$ be the group of units in $\mathbb{Z}_{n}$.
	For $b\in \mathbb{Z}_{n}^{*}$, ${\rm ord}_{n}(b)$ denotes the order of $b$ in $\mathbb{Z}_{n}^{*}$.
	As usual, let $|X|$ denote the cardinality of a finite set $X$.
	For integers $b_{1},b_{2},\cdots,b_{r}$, where $r\geq 2$ is a positive integer,
	${\rm gcd}(b_{1},b_{2},\cdots,b_{r})$ denotes the greatest common divisor of $b_{1},b_{2},\cdots,b_{r}$.
	Given two integers $b_{1}$ and $b_{2}$, if $b_{1}$ divides $b_{2}$  then we write $b_{1} |\, b_{2}$.
	For a positive integer $b$, $\varphi(b)$ is the Euler's function of $b$, which is the number of positive integers not exceeding $b$ and coprime to $b$.
	
	We begin by reviewing the notion of group action on a linear code.
	
	\subsection{Group action on a linear code}
	Suppose that a finite group $G$ acts on a nonempty finite set $X$. For each $x\in X$, $Gx=\{gx\,|\,g\in G\}$ is called an orbit of this group action containing $x$ (or simply a $G$-orbit). All the $G$-orbits partition $X$, that is, $X$ is the disjoint union of the $G$-orbits. For convenience, the set of all the orbits of $G$ on $X$ is denoted as
	$G\backslash X=\{Gx\,|\,x\in X\}.$
	
	Let $\mathcal{C}$ be a linear code and let $\mathcal{G}$ be a subgroup of ${\rm Aut}(\mathcal{C})$. The next lemma reveals that the number of non-zero weights of $\mathcal{C}$ is bounded from above by the number of $\mathcal{G}$-orbits of $\mathcal{C}^{*}=\mathcal{C}\backslash \{\bf 0\}$, with equality if and only if any two codewords of $\mathcal{C}^{*}$ with the same weight are in the same $\mathcal{G}$-orbit.
	
	\begin{lem}{\rm(\cite[Proposition II.2]{7})}\label{l2.1}
		Let $\mathcal{C}$ be a linear code of length $n$ over $\mathbb{F}_{q}$ with $\ell$ non-zero weights and let ${\rm Aut}(\mathcal{C})$ be the automorphism group of $\mathcal{C}$. Suppose that $\mathcal{G}$ is a subgroup of ${\rm Aut}(\mathcal{C})$. If the number of $\mathcal{G}$-orbits of $\mathcal{C}^{*}=\mathcal{C}\backslash \{\bf 0\}$ is equal to $N$, then $\ell \leq N$. Moreover, the equality holds if and only if for any two non-zero codewords ${\bf c}_{1},{\bf c}_{2}\in \mathcal{C}^{*}$ with the same weight, there exists an automorphism $A\in \mathcal{G}$ such that $A{\bf c}_{1}={\bf c}_{2}$.
	\end{lem}
	
	\begin{lem}{\rm (Burnside's lemma) \cite[Theorem 2.113]{16}}
		Let $\mathcal{G}$ acts on a nonempty finite set $X$. Then the number of $\mathcal{G}$-orbits of $X$ is equal to
		\begin{equation}\label{e2.1}
			\frac{1}{|\mathcal{G}|}\sum_{g\in \mathcal{G}}|{\rm Fix}(g)|,
		\end{equation}
		where ${\rm Fix}(g)=\{x\in X\,|\,gx=x\}$.
	\end{lem}
	
	Suppose that both $\mathcal{G}$ and $\mathcal{G}'$ are subgroups of ${\rm Aut}(\mathcal{C})$ and that $\mathcal{G}$ is a subgroup of $\mathcal{G}'$.
	It is easy to see that the number of $\mathcal{G}'$-orbits of $\mathcal{C}^{*}$ is less than or equal to the number of $\mathcal{G}$-orbits of $\mathcal{C}^{*}$. This suggests that if we can find a larger subgroup $\mathcal{G}'$ of ${\rm Aut}(\mathcal{C})$ and count the number of $\mathcal{G}'$-orbits of $\mathcal{C}^{*}$, then we can obtain a smaller upper bound on the number of non-zero weights of $\mathcal{C}$.

	\subsection{$\lambda$-constacyclic codes and primitive idempotents}
	The quotient ring $\mathcal{R}^{(q)}_{n,\lambda}:=\mathbb{F}_{q}[x]/\langle x^{n}-\lambda \rangle$ is semi-simple when $\gcd(n, q) = 1$.
	It is well known that every irreducible $\lambda$-constacyclic code of length $n$ over $\mathbb{F}_{q}$ is generated uniquely by a primitive idempotent of $\mathcal{R}^{(q)}_{n,\lambda}$, and that every $\lambda$-constacyclic code of length $n$ over $\mathbb{F}_{q}$ is a direct sum of some irreducible $\lambda$-constacyclic codes of length $n$ over $\mathbb{F}_{q}$.
	Thus each $\lambda$-constacyclic code  of length $n$ over $\mathbb{F}_{q}$ can be generated uniquely by a idempotent of $\mathcal{R}^{(q)}_{n,\lambda}$. This idempotent is called the generating idempotent of the $\lambda$-constacyclic code.
	
	Let $r$ be the order of $\lambda$ in $\mathbb{F}^{*}_q$. Since $r\,|\,q-1$ and $\gcd(n,q)=1$, we have $\gcd(rn,q)=1$. Let $m$ be the order of $q$ in $\mathbb{Z}_{rn}^{*}$, that is, $m$ is the least positive integer such that $rn\,|\,q^m-1$.
	Let $\omega$ be a primitive element of $\mathbb{F}_{q^m}$ such that $\lambda = {\omega}^{\frac{q^m-1}{r}}$.
	Let $\zeta = {\omega}^{\frac{q^m-1}{rn}}$, then $\zeta$ is a primitive $rn$-th root of unity in $\mathbb{F}_{q^m}$ and ${\zeta}^n = \lambda$.
	Therefore, we have\vspace{-0.5ex}
	$$x^n - \lambda = \prod \limits_{i=0}\limits^{n-1}(x-\zeta^{1+ri}).$$
	It is easy to see that the set of all roots of $x^n - \lambda$ in $\mathbb{F}^{*}_{q^m}$ corresponds to a subset $1 + r \mathbb{Z}_{rn}$ of the residue class ring $\mathbb{Z}_{rn}$, which is defined as follows:\vspace{-0.5ex}
	$$1 + r \mathbb{Z}_{rn}=\{1+ri \mid i=0,1,\cdots,n-1\}.$$
	It is easy to see that $\mathbb{Z}^{*}_{rn} \cap (1 + r \mathbb{Z}_{rn})$ is a subgroup of $\mathbb{Z}^{*}_{rn}$.
	
	There is a one-to-one correspondence between the primitive idempotents of $\mathcal{R}^{(q)}_{n,\lambda}$
	and the $q$-cyclotomic cosets modulo $rn$.
	Assume that all the distinct $q$-cyclotomic cosets modulo $rn$ contained in the set $1 + r \mathbb{Z}_{rn}$ are given by \vspace{-1ex}
	\begin{align*}
		\Gamma_{0}&=\{1+r{a}_{0}=1,q,\cdots,q^{k_0-1}=q^{m-1}\},\\
		\Gamma_{1}&=\{1+r{a}_{1},(1+r{a}_{1})q,\cdots,(1+r{a}_{1})q^{k_1-1}\},\\
		\vdots\\
		\Gamma_{s}&=\{1+r{a}_{s},(1+r{a}_{s})q,\cdots,(1+r{a}_{s})q^{k_s-1}\},
	\end{align*}
	where $k_{i}$ is the cardinality of the $q$-cyclotomic coset $\Gamma_{i}$ for $0\leq i\leq s$ with $k_{0}= m$. It is easy to check that $\Gamma_{0}, \Gamma_{1}, \cdots, \Gamma_{s}$ partition the set $1 + r \mathbb{Z}_{rn}$. Hence, $x^n - \lambda$ can be decomposed into $x^n-\lambda=\prod_{i=0}^{s} m_i(x)$, where $m_i(x) = \prod_{j \in \Gamma_{i}}(x-{\zeta}^j)$ is irreducible over $\mathbb{F}_{q}$ and $m_0(x), m_1(x), \cdots , m_s(x)$ are pairwise coprime.\par
	
	The quotient ring $\mathbb{F}_{q^m}[x]/\langle x^{n}-\lambda\rangle$ has exactly $n$ primitive idempotents given by 
	\vspace{-1.25ex} $$e_{1+ri}=\frac{1}{n}\sum_{l=0}^{n-1}\zeta^{-(1+ri)l}x^{l},~~0\leq i\leq n-1.$$
	Moreover, $\mathcal{R}^{(q)}_{n,\lambda}$ has exactly $s+1$ primitive idempotents given by 
	\vspace{-1ex}
	$$\varepsilon_{t}=\sum_{j\in \Gamma_{t}}e_{j},~~0\leq t\leq s,\vspace{-1.25ex}$$
	and $\mathcal{R}^{(q)}_{n,\lambda}$ is the vector space direct sum of the minimal ideals
	$\mathcal{R}^{(q)}_{n,\lambda}\varepsilon_{t}$ for $0\leq t\leq s$, in symbols,
	\vspace{-0.75ex} $$\mathcal{R}^{(q)}_{n,\lambda}=\mathcal{R}^{(q)}_{n,\lambda}\varepsilon_{0}\bigoplus \mathcal{R}^{(q)}_{n,\lambda}\varepsilon_{1}\bigoplus \cdots \bigoplus \mathcal{R}^{(q)}_{n,\lambda}\varepsilon_{s}.\vspace{-0.25ex}$$
	Using the Discrete Fourier Transform, we have, for each $0\leq t\leq s$, \vspace{-0.5ex}
	$$\mathcal{R}^{(q)}_{n,\lambda}\varepsilon_{t}=\Big\{\sum_{j=0}^{k_{t}-1}\big(\sum_{v=0}^{k_{t}-1}c_{v}\zeta^{v (1+r {a}_{t})q^{j}}\big)e_{(1+r {a}_{t})q^{j}}~\Big|~c_{v}\in \mathbb{F}_{q}, 0\leq v \leq k_{t}-1\Big\}.\vspace{-0.5ex}$$

	\subsection{Subgroups of the automorphism group of $\mathcal{R}^{(q)}_{n,\lambda}$}
	Suppose that the cyclic group $\mathbb{F}_{q}^{*}$ is generated by $\xi$.
	It is easy see that both the $\lambda$-constacyclic shift and the scalar multiplication are  $\mathbb{F}_{q}$-vector space automorphisms of $\mathcal{R}^{(q)}_{n,\lambda}=\mathbb{F}_{q}[x]/\langle x^{n}-\lambda\rangle$, denoted by $\rho$ and $\sigma_{\xi}$, respectively:
	\vspace{-0.25ex}
	$$\rho:~\mathcal{R}^{(q)}_{n,\lambda} \rightarrow \mathcal{R}^{(q)}_{n,\lambda}, \hspace{1.5em} \rho\big(\sum_{i=0}^{n-1}f_{i}x^{i}\big)=\sum_{i=0}^{n-1}f_{i}x^{i+1}~({\rm mod}~x^n-\lambda)\vspace{-0.25ex}$$
	and \vspace{-0.25ex}
	$$\sigma_{\xi}:~\mathcal{R}^{(q)}_{n,\lambda} \rightarrow \mathcal{R}^{(q)}_{n,\lambda}, \hspace{1.25em} \sigma_{\xi}\big(\sum_{i=0}^{n-1}f_{i}x^{i}\big)=\sum_{i=0}^{n-1}\xi f_{i}x^{i}~({\rm mod}~x^n-\lambda).$$
	Clearly, the subgroups $\langle \rho \rangle$ and $\langle\sigma_{\xi} \rangle$ of ${\rm Aut}(\mathcal{R}^{(q)}_{n,\lambda})$ are of order $rn$ and $q-1$, respectively.
	It is easy to verify by the definition that for any $\lambda$-constacyclic code $\mathcal{C}$ of length $n$ over $\mathbb{F}_{q}$, $\langle \rho,\sigma_{\xi}\rangle$ is a subgroup of the automorphism group ${\rm Aut}(\mathcal{C})$ of $\mathcal{C}$.
	
	Let $a \in \mathbb{Z}^{*}_{rn} \cap (1 + r \mathbb{Z}_{rn})$. The multiplier $\mu_{a}$ defined on $\mathcal{R}^{(q)}_{n,\lambda}$ denoted by
	\vspace{-0.5ex} $$\hspace{-1ex}\mu_{a}:~\mathcal{R}^{(q)}_{n,\lambda}\rightarrow \mathcal{R}^{(q)}_{n,\lambda},\hspace{1.25em}
	\mu_{a}\big(\sum_{i=0}^{n-1}f_{i}x^{i}\big)=\sum_{i=0}^{n-1}f_{i}x^{ai}~({\rm mod}~x^n-\lambda)\vspace{-0.5ex}$$
	is a ring automorphism of $\mathcal{R}^{(q)}_{n,\lambda}$ (see \cite{24}).
	It's not hard to prove that the subgroup $\langle \mu_{a}\rangle$ of ${\rm Aut}(\mathcal{R}^{(q)}_{n,\lambda})$ generated by $\mu_{a}$ is of order ${\rm ord}_{rn}(a)$, where ${\rm ord}_{rn}(a)$ denotes the order of $a$ in $\mathbb{Z}_{rn}^{*}$.
	It is easy to verify by definition that for any $\lambda$-constacyclic code $\mathcal{C}$ of length $n$ over $\mathbb{F}_{q}$, the multiplier $\mu_{q}\in {\rm Aut}(\mathcal{C})$ and $\langle \mu_{q} \rangle$ is a subgroup of order $m$ of ${\rm Aut}(\mathcal{C})$, where $m={\rm ord}_{rn}(q)$.

	
	\begin{lem}\label{l2.2}
		With the notation given above, then the subgroup $\langle \mu_{a},\rho,\sigma_{\xi}\rangle$ of ${\rm Aut}(\mathcal{R}^{(q)}_{n,\lambda})$ is of order ${\rm ord}_{rn}(a)n(q-1)$, and each element of $\langle \mu_{a},\rho,\sigma_{\xi} \rangle$ can be written uniquely as a product $\mu_{a}^{r_{1}}\rho^{r_{2}}\sigma_{\xi}^{r_{3}}$, where $0\leq r_{1}\leq {\rm ord}_{rn}(a)-1$, $0\leq r_{2}\leq n-1$ and $0\leq r_{3}\leq q-2$.
	\end{lem}
	
	\begin{proof}
		we note that $\rho {\sigma}_{\xi} = {\sigma}_{\xi} \rho$, then $\langle \rho \rangle \langle {\sigma}_{\xi} \rangle= \langle {\sigma}_{\xi} \rangle \langle \rho \rangle$, which implies that $\langle \rho \rangle \langle {\sigma}_{\xi} \rangle$ is a subgroup of ${\rm Aut}(\mathcal{R}^{(q)}_{n,\lambda})$.
		For any polynomial  $f(x)=\sum_{i=0}^{n-1} f_i x^i \in \mathcal{R}^{(q)}_{n,\lambda}$,
		${\rho}^{n}(f(x)) = \sum_{i=0}^{n-1} \lambda f_i x^i = \sum_{i=0}^{n-1} {\xi}^{j} f_i x^i = {\sigma}_{\xi}^{j}(f(x)) \pmod{x^n - \lambda}$ for some $0 \leq j \leq q-2$.
		It follows that ${\rho}^{n} = {\sigma}_{\xi}^{j}$ for some $0 \leq j \leq q-2$, hence $\langle \rho \rangle \cap \langle {\sigma}_{\xi} \rangle = \langle {\rho}^{n} \rangle$.\par
		
		For any $\rho^{a_1} \in \langle \rho \rangle, {\sigma}_{\xi}^{a_2} \in \langle {\sigma}_{\xi} \rangle$ for some $0 \leq a_1 \leq rn-1$ and $0 \leq a_2 \leq q-2$.
		Let $a_1 = k n + n'$ and $a_2+ k j \equiv l \pmod{q-1}$, where $0 \leq k \leq r-1$, $0 \leq n' \leq n-1$ and $0 \leq l \leq q-2$.
		Then $${\rho}^{a_1}{\sigma}_{\xi}^{a_2} = {\rho}^{kn+n'} {\sigma}_{\xi}^{a_2} = {\rho}^{n'} {\sigma}_{\xi}^{a_2+ kj} = {\rho}^{n'} {\sigma}_{\xi}^{l} \in \langle \rho, {\sigma}_{\xi} \rangle.$$
		Thus each element of $\langle \rho, {\sigma}_{\xi} \rangle$ can be written as a product ${\rho}^{r_2} {\sigma}_{\xi}^{r_3}$ for some $0 \leq r_{2} \leq n-1$ and $0 \leq r_{3} \leq q-2$.\par
		
		For any $a \in \langle \rho, {\sigma}_{\xi} \rangle$, $a = {\rho}^{r_2} {\sigma}_{\xi}^{r_3}$ for some $0 \leq r_{2} \leq n-1$ and $0 \leq r_{3} \leq q-2$. If $a$ can also be written as $a = {\rho}^{r_2'} {\sigma}_{\xi}^{r_3'}$, where $0 \leq r_{2}' \leq n-1$ and $0 \leq r_{3}' \leq q-2$, then $\rho^{r_{2}-r_{2}'}=\sigma_{\xi}^{r_{3}'-r_{3}} \in \langle \rho \rangle \cap \langle {\sigma}_{\xi} \rangle = \langle {\rho}^n \rangle$ and $n\,|\,r_{2}-r_{2}'$.
		As $1-n \leq r_{2}-r_{2}' \leq n-1$, we have $r_{2} = r_{2}'$. Similarly $r_{3} = r_{3}'$.
		Thus, each element of $\langle \rho,\sigma_{\xi} \rangle$ can be written uniquely as a product $\rho^{r_{2}}\sigma_{\xi}^{r_{3}}$ for some $0 \leq r_{2} \leq n-1$ and $0 \leq r_{3} \leq q-2$.\par
		
		Similarly, we note that $\mu_{a}\rho=\rho^{a}\mu_{a}$, $\rho\mu_{a}=\mu_{a}\rho^{a^{-1}}$ and $\mu_{a}\sigma_{\xi}=\sigma_{\xi}\mu_{a}$, then $\langle \mu_{a} \rangle\langle\rho,\sigma_{\xi}\rangle=\langle\rho,\sigma_{\xi} \rangle\langle \mu_{a} \rangle$, which implies that $\langle \mu_{a} \rangle\langle\rho,\sigma_{\xi} \rangle$ is a subgroup of ${\rm Aut}(\mathcal{R}^{(q)}_{n,\lambda})$. Thus $\langle \mu_{a},\rho,\sigma_{\xi} \rangle=\langle \mu_{a} \rangle\langle\rho,\sigma_{\xi} \rangle$.
		Suppose $a \in \langle {\mu}_{a} \rangle \cap \langle \rho,{\sigma}_{\xi} \rangle$, then $a = {\mu}_{a}^{r_1} = {\rho}^{r_2} {\sigma}_{\xi}^{r_3}$ for some $0 \leq r_1 \leq {\rm ord}_{rn}(a)-1, 0 \leq r_2 \leq n-1$ and $0 \leq r_3 \leq q-2$. Let $f(x)=1 \in \mathcal{R}^{(q)}_{n,\lambda}$, then $1 \equiv{\mu}_{a}^{r_1}(f(x)) \equiv {\rho}^{r_2} {\sigma}_{\xi}^{r_3} (f(x)) \equiv {\xi}^{r_3} x^{r_2}~({\rm mod}~x^n-\lambda)$, and hence $r_2 = r_3 = 0$. So $\langle {\mu}_{a} \rangle \cap \langle \rho,{\sigma}_{\xi} \rangle = id$, where $id$ is the identity element of ${\rm Aut}(\mathcal{R}^{(q)}_{n,\lambda})$.
		The rest of the proof is similar to the previous one, so we can quickly conclude that each element of $\langle \mu_{a},\rho,\sigma_{\xi} \rangle$ can be written uniquely as a product $\mu_{a}^{r_{1}}\rho^{r_{2}}\sigma_{\xi}^{r_{3}}$ for some $0 \leq r_{1} \leq {\rm ord}_{rn}(a)-1$, $0 \leq r_{2} \leq n-1$ and $0 \leq r_{3} \leq q-2$.
	\end{proof}
	
	
	
	\section{Improved upper bounds}
	Let $\mathcal{C}$ be a simple-root $\lambda$-constacyclic code of length $n$ over $\mathbb{F}_{q}$ and let $\mathcal{G}$ be a subgroup of ${\rm Aut}(\mathcal{C})$. By Lemma \ref{l2.1}, the number of non-zero weights of $\mathcal{C}$ is bounded from above by the number of $\mathcal{G}$-orbits of $\mathcal{C}^{*}=\mathcal{C}\backslash \{\bf 0\}$.
	Zhang and Cao \cite{21} chose $\mathcal{G}=\langle \rho,\sigma_{\xi} \rangle$ and obtained an upper bound on the number of non-zero weights of $\mathcal{C}$ with the constraint of ${\rm gcd}\big( \frac{q-1}{r}, r \big) = 1$ by counting the number of $\langle \rho,\sigma_{\xi} \rangle$-orbits of $\mathcal{C}^{*}$.
	In this paper, we choose $\mathcal{G}$ to be a larger subgroup of ${\rm Aut}(\mathcal{C})$ which contains $\langle \rho,\sigma_{\xi} \rangle$ as a subgroup, and obtain an smaller upper bound on the number of non-zero weights of $\mathcal{C}$ than \cite{21} by counting the number of $\mathcal{G}$-orbits of $\mathcal{C}^{*}$.
	As remarked at the end of subsection 3.1, our results significantly improve the main results in \cite{21}.
	
	\subsection{An improved upper bound on the number of non-zero weights of an irreducible constacyclic code}
	We know from subsection 2.3 that for a $\lambda$-constacyclic code $\mathcal{C}$ of length $n$ over $\mathbb{F}_{q}$, the multiplier $\mu_{q}$, the constacyclic shift $\rho$ and the scalar multiplication $\sigma_{\xi}$ are all automorphisms of the code $\mathcal{C}$.
	We first assume that $\mathcal{C}$ is an irreducible $\lambda$-constacyclic code and we have the following result.
	
	\begin{Theorem}\label{t3.1}
		Let $\mathcal{C}$ be an $[n,k]$ irreducible $\lambda$-constacyclic code over $\mathbb{F}_{q}$.
		Suppose that the generating idempotent element $\varepsilon_{t}$ of $\mathcal{C}$ corresponds to the $q$-cyclotomic coset $\{1+r{a}_t,(1+r{a}_t)q,\cdots, (1+r{a}_t)q^{k-1}\}$. Then the number of $\langle \mu_{q},\rho,\sigma_{\xi} \rangle$-orbits of $\mathcal{C}^{*}=\mathcal{C}\backslash \{\bf 0\}$ is equal to
		$$\frac{1}{k}\sum_{h \mid k}\varphi\big(\frac{k}{h}\big)
		{\rm gcd}\Big(q^{h}-1,\frac{q^{k}-1}{q-1},\frac{(1+r{a}_t)(q^{k}-1)}{rn}\Big).$$
		In particular, the number of non-zero weights of $\mathcal{C}$ is less than or equal to the number of $\langle \mu_{q},\rho,\sigma_{\xi} \rangle$-orbits of $\mathcal{C}^{*}$, with equality if and only if for any two codewords ${\bf c}_{1},{\bf c}_{2}\in \mathcal{C}^{*}$ with the same weight, there exist integers $j_{1}$, $j_{2}$ and $j_{3}$ such that $\mu_{q}^{j_{1}}\rho^{j_{2}}(\xi^{j_{3}}{\bf c}_{1})={\bf c}_{2}$,
		where $0 \leq j_{1} \leq m-1$, $0 \leq j_{2} \leq n-1$ and $0 \leq j_{3} \leq q-2$.
	\end{Theorem}
	
	\begin{proof}
		It is enough to count the number of $\langle \mu_{q},\rho,\sigma_{\xi} \rangle$-orbits of  $\mathcal{C}^{*}$, since the rest of the statements are clear from Lemma \ref{l2.1}. It follows from Equation (\ref{e2.1}) and Lemma \ref{l2.2} that
		$$\big|\langle \mu_{q},\rho,\sigma_{\xi} \rangle\backslash \mathcal{C}^{*}\big|=\frac{1}{mn(q-1)}\sum_{r_{1}=0}^{m-1}\sum_{r_{2}=0}^{n-1}\sum_{r_{3}=0}^{q-2}\big|{\rm Fix}\big(\mu_{q}^{r_{1}}\rho^{r_{2}}\sigma_{\xi}^{r_{3}}\big)\big|,$$
		where ${\rm Fix}\big(\mu_{q}^{r_{1}}\rho^{r_{2}}\sigma_{\xi}^{r_{3}}\big)=\big\{{\bf c}\in \mathcal{C}^{*}
		\,\big|\,\mu_{q}^{r_{1}}\rho^{r_{2}}\sigma_{\xi}^{r_{3}}({\bf c})={\bf c}\big\}$.
		
		Take a typical non-zero element
		${\bf c}\!=\!\sum_{j=0}^{k-1}
		\big(\sum_{v=0}^{k-1}c_{v}
		\zeta^{v(1+r{a}_t)q^{j}}\big)e_{(1+r{a}_t)q^{j}}\!\in\! \mathcal{C}^{*}.$
		Note that $e_{(1+r{a}_t)q^{j}}\!=\!\frac{1}{n}\sum\limits_{l=0}^{n-1}\zeta^{-(1+r{a}_t)q^{j}l}x^{l}$ and $\rho^{r_{2}}\sigma_{\xi}^{r_{3}}(e_{(1+r{a}_t)q^{j}})=\xi^{r_{3}}\zeta^{(1+r{a}_t)q^{j}r_{2}}e_{(1+r{a}_t)q^{j}}$ (see \cite{21}). It follows that
		\begin{small}	\begin{align*}
				\mu_{q}^{r_{1}}\rho^{r_{2}}\sigma_{\xi}^{r_{3}}(e_{(1+r{a}_t)q^{j}})
				&=\xi^{r_{3}}\zeta^{(1+r{a}_t)q^{j}r_{2}}\mu_{q}^{r_{1}}(e_{(1+r{a}_t)q^{j}})\\
				&=\xi^{r_{3}}\zeta^{(1+r{a}_t)q^{j}r_{2}}\cdot \frac{1}{n}\sum_{l=0}^{n-1}\zeta^{-(1+r{a}_t)q^{j}l}x^{q^{r_{1}}l}\\
				&=\xi^{r_{3}}\zeta^{(1+r{a}_t)q^{j}r_{2}}\cdot \frac{1}{n}\sum_{l=0}^{n-1}\zeta^{-(1+r{a}_t)q^{j-r_{1}}q^{r_{1}}l}x^{q^{r_{1}}l}\\
				&=\xi^{r_{3}}\zeta^{(1+r{a}_t)q^{j}r_{2}}\cdot \frac{1}{n}\sum_{l=0}^{n-1}\zeta^{-(1+r{a}_t)q^{j-r_{1}}l}x^{l}\\
				&=\xi^{r_{3}}\zeta^{(1+r{a}_t)q^{j}r_{2}}e_{(1+r{a}_t)q^{j-r_{1}}},
			\end{align*}
		\end{small}
		where the subscript $(1+r{a}_t)q^{j-r_{1}}$ is calculated modulo $rn$ and the fourth equality holds because $\zeta^{-(1+r{a}_t)q^{j-r_{1}}n}x^n = 1$. Then we have
		\begin{small}		\begin{align*}
				\mu_{q}^{r_{1}}\rho^{r_{2}}\sigma_{\xi}^{r_{3}}({\bf c})
				&=\mu_{q}^{r_{1}}\rho^{r_{2}}\sigma_{\xi}^{r_{3}}\Big(\sum_{j=0}^{k-1}\big(\sum_{v=0}^{k-1}c_{v}\zeta^{v(1+r{a}_t)q^{j}}\big)e_{(1+r{a}_t)q^{j}}\Big)\\ &=\sum_{j=0}^{k-1}\Big(\sum_{v=0}^{k-1}c_{v}\zeta^{v(1+r{a}_t)q^{j}}\Big)\mu_{q}^{r_{1}}\rho^{r_{2}}\sigma_{\xi}^{r_{3}}(e_{(1+r{a}_t)q^{j}})\\			&=\sum_{j=0}^{k-1}\xi^{r_{3}}\zeta^{(1+r{a}_t)q^{j}r_{2}}\Big(\sum_{v=0}^{k-1}c_{v}\zeta^{v(1+r{a}_t)q^{j}}\Big)e_{(1+r{a}_t)q^{j-r_{1}}}\\			&=\sum_{j=0}^{k-1}\xi^{r_{3}}\zeta^{(1+r{a}_t)q^{j-r_{1}}q^{r_{1}}r_{2}}\Big(\sum_{v=0}^{k-1}c_{v}\zeta^{v(1+r{a}_t)q^{j}}\Big)^{q^{r_{1}}}e_{(1+r{a}_t)q^{j-r_{1}}}\\ &=\sum_{j=0}^{k-1}\xi^{r_{3}}\zeta^{(1+r{a}_t)q^{r_{1}+j}r_{2}}\Big(\sum_{v=0}^{k-1}c_{v}\zeta^{v(1+r{a}_t)q^{j}}\Big)^{q^{r_{1}}}e_{(1+r{a}_t)q^{j}}.
			\end{align*}
		\end{small}
		Hence
		$\mu_{q}^{r_{1}}\rho^{r_{2}}\sigma_{\xi}^{r_{3}}({\bf c})\!=\!{\bf c}$ if and only if $\xi^{r_{3}}\big(\sum_{v=0}^{k-1}c_{v}\zeta^{v(1+r{a}_t)q^{j}}\big)^{q^{r_{1}}-1}\!=\!\zeta^{-(1+r{a}_t)q^{r_{1}+j}r_{2}}$
		for $j=0,\!1,\!\cdots,\!k-1$,\! which is equivalent to $\xi^{r_{3}}\big(\sum_{v=0}^{k-1}c_{v}\zeta^{v(1+r{a}_t)}\big)^{q^{r_{1}}-1}=\zeta^{-(1+r{a}_t)q^{r_{1}}r_{2}}.$
		Since the minimal polynomial of $\zeta^{1+r{a}_t}$ over $\mathbb{F}_{q}$ is of degree $k$, the set
		$$\big\{c_{0}+c_{1}\zeta^{1+r{a}_t}+\cdots+c_{k-1}\zeta^{(k-1)(1+r{a}_t)}~\big|~c_{v}\in \mathbb{F}_{q}, 0 \leq v \leq k-1\big\}$$
		forms a subfield of $\mathbb{F}_{q^{m}}$ of size $q^{k}$. It follows that
		$$\big|{\rm Fix}\big(\mu_{q}^{r_{1}}\rho^{r_{2}}\sigma_{\xi}^{r_{3}}\big)\big|=\Big|\big\{\alpha\in \mathbb{F}_{q^{k}}^{*}~\big|~ \xi^{r_{3}}\alpha^{q^{r_{1}}-1}=\zeta^{-(1+r{a}_t)q^{r_{1}}r_{2}}\big\}\Big|.$$
		
		Suppose that
		$\xi^{r_{3}}\alpha^{q^{r_{1}}-1}=\zeta^{-(1+r{a}_t)q^{r_{1}}r_{2}}$ for some $\alpha\in \mathbb{F}_{q^{k}}^{*}$.
		Let $\mathbb{F}_{q^{k}}^{*}$ be generated by $\theta$.
		On the one hand, for any $\beta \!\in\! \big\langle \theta^{\frac{q^{k}-1}{q^{{\rm gcd}(k,r_{1})}-1}} \big\rangle$,
		we have $\xi^{r_{3}}(\alpha\beta)^{q^{r_{1}}-1}\!=\!
		\xi^{r_{3}}\alpha^{q^{r_{1}}-1}\!=\!\zeta^{-(1+r{a}_t)q^{r_{1}}r_{2}}$.
		On the other hand, suppose $\gamma\!\in\! \mathbb{F}_{q^{k}}^{*}$ satisfies $\xi^{r_{3}}\gamma^{q^{r_{1}}-1}\!=\!\zeta^{-(1+r{a}_t)q^{r_{1}}r_{2}}$,
		then $(\gamma \alpha^{-1})^{q^{r_{1}}-1}=1$, and so ${\rm ord}(\gamma \alpha^{-1})\,\big|\,{\rm gcd}(q^{k}-1,q^{r_{1}}-1)=q^{{\rm gcd}(k,r_{1})}-1$, which implies that $\gamma \alpha^{-1}\!\in\!\big\langle \theta^{\frac{q^{k}-1}{q^{{\rm gcd}(k,r_{1})}-1}} \big\rangle$. Hence $\gamma\!=\!\alpha\beta'$ for some $\beta' \!\in\! \big\langle \theta^{\frac{q^{k}-1}{q^{{\rm gcd}(k,r_{1})}-1}} \big\rangle$.
		It follows that\vspace{-0.5ex}
		\begin{align*}
			\big|{\rm Fix}\big(\mu_{q}^{r_{1}}\rho^{r_{2}}
			\sigma_{\xi}^{r_{3}}\big)\big|=0~{\rm or}~q^{{\rm gcd}(k,r_{1})}-1.
		\end{align*}
		Next, fixing $r_1~(0 \leq\! r_1 \!\leq m-1)$, we count the number of number pairs $(r_2,r_3)$ such that $\big|{\rm Fix}\big(\mu_{q}^{r_{1}}\rho^{r_{2}}\sigma_{\xi}^{r_{3}}\big)\big|\!\neq\!0$ , where $0 \leq r_{2} \leq n-1$, $0 \leq r_{3} \leq q-2$.
		
		It's not hard to see that $\{\alpha^{q^{r_{1}}-1}\,|\,\alpha\in \mathbb{F}_{q^{k}}^{*}\}=\langle \theta^{q^{r_{1}}-1}\rangle$, which is a cyclic subgroup of $\mathbb{F}_{q^{k}}^{*}$ of order $\frac{q^{k}-1}{q^{{\rm gcd}(k,r_{1})}-1}$.
		Since $\langle \xi \rangle$ is a cyclic subgroup of $\mathbb{F}_{q^{k}}^{*}$ of order $q-1$,
		we see that $\langle \xi \rangle \cap \langle \theta^{q^{r_{1}}-1} \rangle$ is a cyclic subgroup of $\mathbb{F}_{q^{k}}^{*}$ of order ${\rm gcd}\big(q-1,\frac{q^{k}-1}{q^{{\rm gcd}(k,r_{1})}-1}\big)$ and $\langle \xi \rangle \langle \theta^{q^{r_{1}}-1} \rangle$ is a cyclic subgroup of $\mathbb{F}_{q^{k}}^{*}$ of order $\frac{|\langle \xi \rangle||\langle \theta^{q^{r_{1}}-1} \rangle|}{|\langle \xi \rangle \cap \langle \theta^{q^{r_{1}}-1} \rangle|}$. As ${\rm ord}(\zeta^{-1})= rn$, ${\rm ord}(\zeta^{-(1+r{a}_t)q^{r_{1}}r_{2}})=\frac{rn}{{\rm gcd}(rn,(1+r{a}_t)q^{r_{1}}r_{2})}=\frac{rn}{{\rm gcd}(rn,(1+r{a}_t)r_{2})}$.
		Then we have
		\begin{small}		\begin{align*}
				\zeta^{-(1+r{a}_t)q^{r_{1}}r_{2}}\in \langle \xi \rangle\langle \theta^{q^{r_{1}}-1}\rangle~
				&\Leftrightarrow~\frac{rn}{{\rm gcd}(rn,(1+r{a}_t)r_{2})}\,\Big|\, |\langle \xi \rangle\langle \theta^{q^{r_{1}}-1} \rangle|\\
				&\Leftrightarrow~\frac{rn}{{\rm gcd}(rn,(1+r{a}_t)r_{2})\!\cdot\!{\rm gcd}\big(\frac{rn}{{\rm gcd}(rn,(1+r{a}_t)r_{2})},|\langle \xi \rangle\langle \theta^{q^{r_{1}}-1} \rangle|\big)}=1\\
				&\Leftrightarrow~\frac{rn}{{\rm gcd}\big(rn,(1+r{a}_t)r_{2}|\langle \xi \rangle\langle \theta^{q^{r_{1}}-1} \rangle|\big)}=1\\
				&\Leftrightarrow~rn\,\big|\,\big((1+r{a}_t)r_{2}|\langle \xi \rangle\langle \theta^{q^{r_{1}}-1} \rangle|\big)\\
				&\Leftrightarrow~\frac{rn}{{\rm gcd}\big(rn,(1+r{a}_t)|\langle \xi \rangle\langle \theta^{q^{r_{1}}-1} \rangle|\big)}\,\Big|\,r_{2}\\
				&\Leftrightarrow~\frac{n}{{\rm gcd}\big(n,\frac{(1+r{a}_t)|\langle \xi \rangle\langle \theta^{q^{r_{1}}-1} \rangle|}{r}\big)}\,\Big|\,r_{2}.
			\end{align*}
		\end{small}
		Let
		$$S(r_{1})=\big\{0 \leq z \leq n-1\,\big|\, \zeta^{-(1+r{a}_t)q^{r_{1}}z}\in \langle \xi \rangle\langle \theta^{q^{r_{1}}-1}\rangle\big\}.$$
		It follows from the above discussion that
		$|S(r_{1})|={\rm gcd}\big(n,\frac{(1+r{a}_t)|\langle \xi \rangle \langle \theta^{q^{r_{1}}-1} \rangle|}{r}\big).$
		Assume that $r_{2}\in S(r_{1})$, then there exists $r_{3}~(0 \leq r_{3}\leq q-2)$ such that
		$\zeta^{-(1+r{a}_t)q^{r_{1}}r_{2}}\in \xi^{r_{3}}\langle \theta^{q^{r_{1}}-1}\rangle$. Denote
		$$R(r_{1},r_{2})=\big\{0 \leq z \leq q-2\,\big|\,\zeta^{-(1+r{a}_t)q^{r_{1}}r_{2}}\in \xi^{z}\langle \theta^{q^{r_{1}}-1}\rangle\big\}.$$
		On the one hand, for any $\xi^{z'}\in \langle \xi \rangle \cap \langle \theta^{q^{r_{1}}-1} \rangle$, we have $\zeta^{-(1+r{a}_t)q^{r_{1}}r_{2}}\in \xi^{r_{3}}\langle \theta^{q^{r_{1}}-1}\rangle=\xi^{r_{3}+z'}\langle \theta^{q^{r_{1}}-1}\rangle$, and so $r_{3}+z'~({\rm mod}~q-1)\in R(r_{1},r_{2})$. On the other hand, suppose $z\in R(r_{1},r_{2})$, that is, $\zeta^{-(1+r{a}_t)q^{r_{1}}r_{2}}\in \xi^{z}\langle \theta^{q^{r_{1}}-1}\rangle$, then $\xi^{z-r_{3}}\in \langle \xi \rangle \cap \langle \theta^{q^{r_{1}}-1}\rangle$, and so $z \equiv r_{3}+z'~({\rm mod}~q-1)$, where $z'$ is an integer such that $\xi^{z-r_{3}}=\xi^{z'}\in \langle \xi \rangle \cap \langle \theta^{q^{r_{1}}-1}\rangle$. Hence
		$|R(r_{1},r_{2})|=|\langle \xi \rangle \cap \langle \theta^{q^{r_{1}}-1}\rangle|.$

		To sum up, we have
		\begin{small}		\begin{align*}
				&\big|\langle \mu_{q},\rho,\sigma_{\xi} \rangle \backslash \mathcal{C}^{*}\big|\\
				=&\frac{1}{mn(q-1)}\sum_{r_{1}=0}^{m-1}\sum_{r_{2}\in S(r_{1})}\sum_{r_{3}\in R(r_{1},r_{2})}(q^{{\rm gcd}(k,r_{1})}-1)\\
				=&\frac{1}{mn(q-1)}\sum_{r_{1}=0}^{m-1}|S(r_{1})|\!\cdot\!|R(r_{1},r_{2})|(q^{{\rm gcd}(k,r_{1})}-1)\\
				=&\frac{1}{mn(q-1)}\sum_{r_{1}=0}^{m-1}{\rm gcd}\Big(n|\langle \xi \rangle \cap \langle \theta^{q^{r_{1}}-1}\rangle|,\frac{(1+r{a}_t)|\langle \xi \rangle||\langle \theta^{q^{r_{1}}-1} \rangle|}{r}\Big)(q^{{\rm gcd}(k,r_{1})}-1)\\
				=&\frac{1}{mn(q-1)}\sum_{r_{1}=0}^{m-1}{\rm gcd}\Big(n(q-1),\frac{n(q^{k}-1)}{q^{{\rm gcd}(k,r_{1})}-1},\frac{(1+r{a}_t)(q-1)(q^{k}-1)}{r(q^{{\rm gcd}(k,r_{1})}-1)}\Big)(q^{{\rm gcd}(k,r_{1})}-1)\\
				=&\frac{1}{mn(q-1)}\sum_{r_{1}=0}^{m-1}{\rm gcd}\Big(n(q-1)(q^{{\rm gcd}(k,r_{1})}-1),n(q^{k}-1),\frac{(1+r{a}_t)(q-1)(q^{k}-1)}{r}\Big)\\
				=&\frac{1}{m}\sum_{r_{1}=0}^{m-1}{\rm gcd}\Big(q^{{\rm gcd}(k,r_{1})}-1,\frac{q^{k}-1}{q-1},\frac{(1+r{a}_t)(q^{k}-1)}{rn}\Big)\\
				=&\frac{1}{m}\!\cdot\! \frac{m}{k}\sum_{r_{1}=0}^{k-1}{\rm gcd}\Big(q^{{\rm gcd}(k,r_{1})}-1,\frac{q^{k}-1}{q-1},\frac{(1+r{a}_t)(q^{k}-1)}{rn}\Big)\\
				=&\frac{1}{k}\sum_{h\mid k}\varphi\big(\frac{k}{h}\big){\rm gcd}\Big(q^{h}-1,\frac{q^{k}-1}{q-1},\frac{(1+r{a}_t)(q^{k}-1)}{rn}\Big).
			\end{align*}
		\end{small}
		We have completed the proof of the theorem.
	\end{proof}
	
	\begin{Remark}\label{r.1}{\rm
			(1)~Let $\mathcal{C}$ be the irreducible $\lambda$-constacyclic code in Theorem \ref{t3.1}. In \cite[Lemma 3.5]{21}, the authors assumed ${\rm gcd}\big( \frac{q-1}{r}, r \big) = 1$ and gave the number of $\langle \rho,\sigma_{\xi} \rangle$-orbits of $\mathcal{C}^{*}=\mathcal{C}\backslash \{\bf 0\}$ as follows:
			\vspace{-0.25ex}
			\begin{equation}\label{E.1}
				{\rm gcd}\Big(\frac{q^{k}-1}{q-1},\frac{(1+r{a}_t)(q^{k}-1)}{rn}\Big).
			\end{equation}
			Indeed, (3.1) still holds when the restriction ${\rm gcd}\big( \frac{q-1}{r}, r \big) = 1$ is removed. This can be verified by letting $r_1=0$ in the proof of Theorem \ref{t3.1}.
			
			(2)~ By Theorem \ref{t3.1} and Equation (\ref{E.1}), we have
			\begin{small}		\begin{align*}
					&\big|\langle \rho,\sigma_{\xi} \rangle \backslash \mathcal{C}^{*}\big|-\big|\langle \mu_{q},\rho,\sigma_{\xi} \rangle \backslash \mathcal{C}^{*}\big|\\
					=&{\rm gcd}\Big(\frac{q^{k}-1}{q-1},\frac{(1+r{a}_t)(q^{k}-1)}{rn}\Big)
					-\frac{1}{k}\sum_{h\mid k}\varphi\big(\frac{k}{h}\big){\rm gcd}\Big(q^{h}-1,\frac{q^{k}-1}{q-1},\frac{(1+r{a}_t)(q^{k}-1)}{rn}\Big)\\
					=&\frac{1}{k}\sum_{h\mid k}\varphi\big(\frac{k}{h}\big)\left({\rm gcd}\Big(\frac{q^{k}-1}{q-1},\frac{(1+r{a}_t)(q^{k}-1)}{rn}\Big)-{\rm gcd}\Big(q^{h}-1,\frac{q^{k}-1}{q-1},\frac{(1+r{a}_t)(q^{k}-1)}{rn}\Big)\right).
				\end{align*}
			\end{small}
			It is easy to see that $\big|\langle \mu_{q},\!\rho,\!\sigma_{\xi} \rangle \backslash \mathcal{C}^{*}\big|
			\!\leq\! \big|\langle \rho,\!\sigma_{\xi} \rangle \backslash \mathcal{C}^{*}\big|$,
			\!with equality if and only if
			${\rm gcd}\big(q\!-\!1,\!\frac{q^{k}\!-\!1}{q\!-\!1},\!\frac{(1\!+\!r{a}_t)(q^{k}\!-\!1)}{rn}\big)\!=\!{\rm gcd}\big(\frac{q^{k}-1}{q-1},\frac{(1+r{a}_t)(q^{k}-1)}{rn}\big)$.
			It follows from Theorem \ref{t3.1} that \vspace{-0.5ex}
			\begin{align*}
				\big|\langle \mu_{q},\rho,\sigma_{\xi} \rangle \backslash \mathcal{C}^{*}\big|\geq \Big\lceil \frac{1}{k}\big(k-1+\big|\langle \rho,\sigma_{\xi} \rangle \backslash \mathcal{C}^{*}\big|\big)\Big\rceil
			\end{align*}
			and \vspace{-0.75ex}
			\begin{align*}		
				0 \leq \big|\langle \rho,\sigma_{\xi} \rangle \backslash \mathcal{C}^{*} \big|
				-\big|\langle \mu_{q},\rho,\sigma_{\xi} \rangle \backslash \mathcal{C}^{*}\big|
				\leq \Big\lfloor
				\frac{(k-1)\big(\big|\langle \rho,\sigma_{\xi} \rangle \backslash \mathcal{C}^{*}\big|-1\big)}{k}\Big\rfloor,
			\end{align*}
			where for a rational number $x$, $\lceil x\rceil$ denotes the smallest integer greater than or equal to $x$ and $\lfloor x\rfloor$ denotes the largest integer less than or equal to $x$.
			If $k>1$ is a prime,  $a_t=0$, $n=\frac{q^k-1}{rN}$ with $N>1$ and ${\rm gcd}(q-1,N)=1$,
			then $\big|\langle \rho,\sigma_{\xi} \rangle \backslash \mathcal{C}^{*}\big|-\big|\langle \mu_{q},\rho,\sigma_{\xi} \rangle \backslash \mathcal{C}^{*}\big|$ meets the above upper bound,
			which is equal to $\frac{(k-1)(N-1)}{k}$.}
	\end{Remark}
	
	We present an example to show that the upper bound in Theorem \ref{t3.1} is tight and in some cases strictly smaller than the one given in Remark \ref{r.1} (1), i.e., Equation (\ref{E.1}).
	
	\begin{Example}{\rm
			Take $q=3$, $n=8$ and $\lambda=-1$. All the distinct $3$-cyclotomic cosets modulo $16$ are given by
			$\Gamma_{0}=\{1,3,9,11\},~\Gamma_{1}=\{5,7,13,15\}.$
			Consider the irreducible negacyclic code
			$\mathcal{C}=\mathcal{R}^{(q)}_{n,\lambda}\varepsilon_{1}$,
			where the primitive idempotent $\varepsilon_{1}$ corresponds to $\Gamma_{1}$.
			Let $\ell$ be the number of non-zero weights of $\mathcal{C}$.
			By Equation (\ref{E.1}), we obtain $\ell \leq 5$.
			Using Theorem \ref{t3.1}, we have $\ell\leq 2$.
			Using the Magma software programming \cite{4},
			we see that the weight enumerator of $\mathcal{C}$ is $1+16x^3 +64x^6$, which implies that $\ell=2$.
			Then Theorem \ref{t3.1} ensures that all the non-zero codewords of $\mathcal{C}$ with the same weight are in the same $\langle \mu_{q},\rho,\sigma_{\xi} \rangle$-orbit.}
	\end{Example}
	
	The following two corollaries show that Theorem \ref{t3.1} can be used to construct some new few-weight
	irreducible $\lambda$-constacyclic codes.
	
	\begin{Corollary}\label{c3.1}
		Suppose that $q=2^{m'}$, where $m'>1$ is odd.
		Let $\mathcal{C}$ be an irreducible $\lambda$-constacyclic code of length $n$ over $\mathbb{F}_{q}$ whose generating idempotent $\varepsilon_{t}$ corresponds to the $q$-cyclotomic coset $\{1+r{a}_t,(1+r{a}_t)q,(1+r{a}_t)q^2,\cdots\}$.
		If $n=\frac{q^{2}-1}{3rN}$, where $N\,|\,\frac{q-1}{r}$, and
		${\rm gcd}\big(1+r{a}_t, \frac{q+1}{3}\big)=1$, then $\mathcal{C}$ is a one-weight or two-weight $\lambda$-constacyclic code.
		
	\end{Corollary}
	
	\begin{proof}
		Let $k_{t}$ be the least positive integer such that $(1+r{a}_t)q^{k_{t}}\equiv 1+r{a}_t \pmod{rn}$, equivalently, $q^{k_{t}}\equiv 1~\big({\rm mod}~\frac{rn}{{\rm gcd}(rn,1+r{a}_t)}\big)$.
		Note that $q^{2}\equiv 1~\big({\rm mod}~\frac{rn}{{\rm gcd}(rn,1+r{a}_t)}\big)$, then $k_{t}\,|\,2$, and so $k_{t}=1$ or $2$. If $k_{t}=1$, then $rn\,|\,(1+r{a}_t)(q-1)$, that is, $\frac{(q+1)(q-1)}{3N}\,|\,(1+r{a}_t)(q-1)$, implying $\frac{q+1}{3}\,|\,(1+r{a}_t)N$.
		Since ${\rm gcd}\big(1+r{a}_t,\frac{q+1}{3})=1$, we have $\frac{q+1}{3}\,|\,N$, which is impossible because $N\,|\,(q-1)$, $\frac{q+1}{3}\,|\,(q+1)$ and ${\rm gcd}(q-1,q+1)=1$. So $k_{t}=2$. It follows from Theorem \ref{t3.1} that
		\begin{small}
			\begin{align*}
				\big|\langle \mu_{q},\rho,\sigma_{\xi} \rangle \backslash \mathcal{C}^{*}\big|
				=&\frac{1}{2}\sum_{h\mid 2}\varphi(\frac{2}{h}){\rm gcd}\Big(q^{h}-1,\frac{q^{2}-1}{q-1},\frac{(1+r{a}_t)(q^{2}-1)}{rn}\Big)\\
				=&\frac{1}{2}\left[\varphi(2){\rm gcd}\Big(q-1,q+1,\frac{(1+r{a}_t)(q^{2}-1)}{rn}\Big)
				+\varphi(1){\rm gcd}\Big(q^{2}-1,q+1,\frac{(1+r{a}_t)(q^{2}-1)}{rn}\Big)\right]\\
				=&\frac{1}{2}\left[\varphi(2) +\varphi(1)\cdot {\rm gcd}\Big(q+1,\frac{(1+r{a}_t)(q^{2}-1)}{rn}\Big)\right]\\
				=&\frac{1}{2}\Big[1 + {\rm gcd}\big(q+1,3N(1+r{a}_t)\big)\Big]\\
				=&\frac{1}{2}(1+3)=2.
			\end{align*}
		\end{small}
		Therefore, the number of non-zero weights of $\mathcal{C}$ is less than or equal to $2$, that is, $\mathcal{C}$ is a one-weight or two-weight $\lambda$-constacyclic code.
	\end{proof}
	
	\begin{Example}{\rm
			Take $q=32$, $n=11$, $\lambda=\theta$,
			where $\theta$ is a primitive element of $\mathbb{F}_{32}$.
			Consider the irreducible $\lambda$-constacyclic code $\mathcal{C}=\mathcal{R}^{(q)}_{n,\lambda}\varepsilon_{0}$, where the primitive idempotent $\varepsilon_{0}$ corresponds to the $32$-cyclotomic coset $\Gamma_{0}=\{1, 32\}$.
			According to Corollary \ref{c3.1}, $\mathcal{C}$ is a one-weight or two-weight $\lambda$-constacyclic code.
			In fact, by use of Magma \cite{4}, the weight enumerator of $\mathcal{C}$ is $1+341x^{10}+682x^{11}$, that is, $\mathcal{C}$ is a two-weight $\lambda$-constacyclic code. Therefore, Theorem \ref{t3.1} guarantees that any two codewords of $\mathcal{C}$ with the same weight are in the same $\langle\mu_{q},\rho,\sigma_{\xi} \rangle$-orbit.}
	\end{Example}
	
	\begin{Corollary}\label{c3.2}
		Suppose that $(q,k)\neq (2,3)$ and that $k$ and $2k+1$ are odd primes satisfying ${\rm gcd}(q-1,k)=1$, ${\rm gcd}(q-1,2k+1)=1$ and $q^{k}\equiv 1~({\rm mod}~2k+1)$.
		Let $\mathcal{C}$ be an irreducible $\lambda$-constacyclic code of length $n$ over $\mathbb{F}_{q}$ whose generating idempotent $\varepsilon_{t}$ corresponds to the $q$-cyclotomic coset $\{1+r{a}_t,(1+r{a}_t)q,(1+r{a}_t)q^2,\cdots\}$.
		If $n=\frac{q^{k}-1}{(2k+1)rN}$, where $N\,|\,\frac{q-1}{r}$,
		and ${\rm gcd}\big(1+r{a}_t, \frac{q^{k}-1}{(2k+1)(q-1)}\big)=1$, then $\mathcal{C}$ is a one-, two- or three-weight $\lambda$-constacyclic code.
	\end{Corollary}
	
	\begin{proof}
		The proof is similar to that of Corollary \ref{c3.1}, and we omit it here.
	\end{proof}

	\begin{Example}{\rm
			Take $q=3$, $n=11$ and $\lambda=-1$.
			Consider the irreducible negacyclic code $\mathcal{C}=\mathcal{R}_{n,\lambda}^{(q)}\varepsilon_{1}$, where the primitive idempotent $\varepsilon_{1}$ corresponds to the $3$-cyclotomic coset $\Gamma_{1}=\{7,13,17,19,21\}$. According to Corollary \ref{c3.2}, $\mathcal{C}$ is a one-, two- or three-weight $\lambda$-constacyclic code. By use of Magma \cite{4}, the weight enumerator of $\mathcal{C}$ is $1+132x^{6}+110x^{9}$, that is, $\mathcal{C}$ is a two-weight $\lambda$-constacyclic code.}
	\end{Example}
	
	\subsection{An improved upper bound on the number of non-zero weights of a general constacyclic code}
	We now turn to consider the action of $\langle\mu_{q},\rho,\sigma_{\xi} \rangle$ on a general $\lambda$-constacyclic code $\mathcal{C}$. Let $j_{1},j_{2},\cdots,j_{u}$ be positive integers and let $t_{j_{1}},t_{j_{2}},\cdots,t_{j_{u}}$ be integers with $0\leq t_{j_{1}}<t_{j_{2}}<\cdots<t_{j_{u}}\leq s$.
	Suppose that the irreducible $\lambda$-constacyclic code $\mathcal{R}^{(q)}_{n,\lambda}\varepsilon_{t_{j_{i}}}$ corresponds to the $q$-cyclotomic coset $\{1+ra_{t_{j_{i}}},(1+ra_{t_{j_{i}}})q,\cdots,(1+ra_{t_{j_{i}}})q^{k_{t_{j_{i}}}-1}\}$ for $1\leq i\leq u$. Define
	$$\mathcal{C}_{j_{1},j_{2},\cdots, j_{u}}^{\sharp}=\mathcal{R}^{(q)}_{n,\lambda}\varepsilon_{t_{j_{1}}}\backslash\{{\bf 0}\}\bigoplus \mathcal{R}^{(q)}_{n,\lambda}\varepsilon_{t_{j_{2}}}\backslash\{{\bf 0}\}\bigoplus \cdots \bigoplus \mathcal{R}^{(q)}_{n,\lambda}\varepsilon_{t_{j_{u}}}\backslash\{{\bf 0}\}.$$
	The following lemma gives an upper bound on the number of $\langle\mu_{q},\rho,\sigma_{\xi} \rangle$-orbits of $\mathcal{C}_{j_{1},j_{2},\cdots, j_{u}}^{\sharp}$.
	
	\begin{lem}\label{l3.1}
		With the notation given above, then the number of $\langle\mu_{q},\rho,\sigma_{\xi} \rangle$-orbits of $\mathcal{C}_{j_{1},j_{2},\cdots, j_{u}}^{\sharp}$ is less than or equal to
		\begin{align*}
			&\frac{1}{mn(q-1)}\sum_{h=0}^{m-1}{\rm gcd}\Big(n,\frac{(1+ra_{t_{j_{1}}})II_{t_{j_{1}}}}{r\!\cdot\!{\rm gcd}(I,I_{t_{j_{1}}})},\cdots,\frac{(1+ra_{t_{j_{u}}})II_{t_{j_{u}}}}{r\!\cdot\!{\rm gcd}(I,I_{t_{j_{u}}})},\frac{(a_{t_{j_{2}}}\!-\!a_{t_{j_{1}}})I_{t_{j_{1}}}I_{t_{j_{2}}}}{{\rm gcd}(I_{t_{j_{1}}},I_{t_{j_{2}}})},\cdots,\\
			&\frac{(a_{t_{j_{u}}}\!-\!a_{t_{j_{1}}})I_{t_{j_{1}}}I_{t_{j_{u}}}}{{\rm gcd}(I_{t_{j_{1}}},I_{t_{j_{u}}})},\cdots,
			\frac{(a_{t_{j_{u}}}\!-\!a_{t_{j_{u-1}}})I_{t_{j_{u-1}}}I_{t_{j_{u}}}}{{\rm gcd}(I_{t_{j_{u-1}}},I_{t_{j_{u}}})}\Big)
			\!\cdot\!{\rm gcd}(I,I_{t_{j_{1}}},\cdots,I_{t_{j_{u}}})\!\cdot\!
			\prod_{i=1}^{u}(q^{{\rm gcd}(k_{t_{j_{i}}},h)}-1),
		\end{align*}
		where $I=q-1$, $I_{t_{j_{i}}}=\frac{q^{k_{t_{j_{i}}}}-1}{q^{{\rm gcd}(k_{t_{j_{i}}},h)}-1}$ for $i=1,2,\cdots,u$.
	\end{lem}
	
	\begin{proof}
		According to Equation (\ref{e2.1}) and Lemma \ref{l2.2}, we see that
		$$\big|\langle\mu_{q},\rho,\sigma_{\xi} \rangle\backslash \mathcal{C}_{j_{1},j_{2},\cdots, j_{u}}^{\sharp}\big|=\frac{1}{mn(q-1)}\sum_{r_{1}=0}^{m-1}\sum_{r_{2}=0}^{n-1}\sum_{r_{3}=0}^{q-2}\big|{\rm Fix}\big(\mu_{q}^{r_{1}}\rho^{r_{2}}\sigma_{\xi}^{r_{3}}\big)\big|,$$
		where ${\rm Fix}\big(\mu_{q}^{r_{1}}\rho^{r_{2}}\sigma_{\xi}^{r_{3}}\big)=\big\{{\bf c}\in \mathcal{C}_{j_{1},j_{2},\cdots, j_{u}}^{\sharp}
		\,\big|\,\mu_{q}^{r_{1}}\rho^{r_{2}}\sigma_{\xi}^{r_{3}}({\bf c})={\bf c}\big\}$.
		
		Let ${\bf c}={\bf c}_{t_{j_{1}}}+{\bf c}_{t_{j_{2}}}+\cdots+{\bf c}_{t_{j_{u}}}\in \mathcal{C}_{j_{1},j_{2},\cdots, j_{u}}^{\sharp}$, where ${\bf c}_{t_{j_{i}}}\in \mathcal{R}^{(q)}_{n,\lambda}\varepsilon_{t_{j_{i}}}\backslash\{{\bf 0}\}$ for $i=1,2,\cdots,u$.
		Suppose that
		$${\bf c}_{t_{j_{i}}}=\sum_{j=0}^{k_{t_{j_{i}}}-1}\big(\sum_{v=0}^{k_{t_{j_{i}}}-1} c_{t_{j_{i}},v}\zeta^{v(1+ra_{t_{j_{i}}})q^{j}}\big)e_{(1+ra_{t_{j_{i}}})q^{j}}
		~~{\rm for}~i=1,2,\cdots,u.$$
		Then we have
		$\mu_{q}^{r_{1}}\rho^{r_{2}}\sigma_{\xi}^{r_{3}}({\bf c})
		=\mu_{q}^{r_{1}}\rho^{r_{2}}\sigma_{\xi}^{r_{3}}({\bf c}_{t_{j_{1}}})+\cdots
		+\mu_{q}^{r_{1}}\rho^{r_{2}}\sigma_{\xi}^{r_{3}}({\bf c}_{t_{j_{u}}})$, where
		\begin{align*}
			\mu_{q}^{r_{1}}\rho^{r_{2}}\sigma_{\xi}^{r_{3}}({\bf c}_{t_{j_{i}}})
			=\sum_{j=0}^{k_{t_{j_{i}}}-1}\xi^{r_{3}}\zeta^{(1+ra_{t_{j_{i}}})q^{r_{1}+j}r_{2}}
			\big(\sum_{v=0}^{k_{t_{j_{i}}}-1}c_{t_{j_{i}},v}\zeta^{v(1+ra_{t_{j_{i}}})q^{j}}\big)^{q^{r_{1}}}e_{(1+ra_{t_{j_{i}}})q^{j}}
			~~{\rm for}~i=1,2,\cdots,u.
		\end{align*}
		It follows that
		\begin{align*}
			\mu_{q}^{r_{1}}\rho^{r_{2}}\sigma_{\xi}^{r_{3}}({\bf c})={\bf c}
			&\Leftrightarrow~ \mu_{q}^{r_{1}}\rho^{r_{2}}\sigma_{\xi}^{r_{3}}({\bf c}_{t_{j_{i}}})={\bf c}_{t_{j_{i}}}~{\rm for}~~i=1,2,\cdots,u\\
			&\Leftrightarrow~\xi^{r_{3}}\big(\sum_{v=0}^{k_{t_{j_{i}}}-1}c_{t_{j_{i}},v}\zeta^{v(1+ra_{t_{j_{i}}})}\big)^{q^{r_{1}}-1}
			=\zeta^{-(1+ra_{t_{j_{i}}})q^{r_{1}}r_{2}}~~{\rm for}~i=1,2,\cdots,u.
		\end{align*}
		Since the minimal polynomial of $\zeta^{1+ra_{t_{j_{i}}}}~(1 \leq i \leq u)$ over $\mathbb{F}_{q}$ is of degree $k_{t_{j_{i}}}$, the set
		$$\Big\{c_{t_{j_{i}},0}+c_{t_{j_{i}},1}\zeta^{1+ra_{t_{j_{i}}}}
		+\cdots+c_{t_{j_{i}},k_{t_{j_{i}}}-1}\zeta^{(k_{t_{j_{i}}}-1)(1+ra_{t_{j_{i}}})}~\Big|~
		c_{t_{j_{i}},v}\in \mathbb{F}_{q},~ 1 \leq l \leq k_{t_{j_{i}}}-1\Big\}$$
		forms a subfield of $\mathbb{F}_{q^{m}}$ with $q^{k_{t_{j_{i}}}}$ elements.
		So it turns out that
		\begin{align*}
			&\big|{\rm Fix}\big(\mu_{q}^{r_{1}}\rho^{r_{2}}\sigma_{\xi}^{r_{3}}\big)\big|\\
			=&\Big|\big\{(\alpha_{t_{j_{1}}},\cdots,\alpha_{t_{j_{u}}})\!\in \mathbb{F}_{q^{k_{t_{j_{1}}}}}^{*}\!\times \cdots \times \mathbb{F}_{q^{k_{t_{j_{u}}}}}^{*}\!~\big|~\xi^{r_{3}}\alpha_{t_{j_{i}}}^{q^{r_{1}}-1}=\zeta^{-(1+ra_{t_{j_{i}}})q^{r_{1}}r_{2}},~i=1,2,\cdots,u \big\}\Big|.
		\end{align*}
		It's easy to prove that
		\begin{align*}
			\big|{\rm Fix}\big(\mu_{q}^{r_{1}}\rho^{r_{2}}\sigma_{\xi}^{r_{3}}\big)\big|=0~{\rm or} ~\prod_{i=1}^{u}(q^{{\rm gcd}(k_{t_{j_{i}}},r_{1})}-1).
		\end{align*}
		Next, fixing $r_1~\!(0 \leq \!r_1 \!\leq m-1)$, we count the number of number pairs $(r_2,r_3)$ such that $\big|{\rm Fix}\big(\mu_{q}^{r_{1}}\rho^{r_{2}}\sigma_{\xi}^{r_{3}}\big)\big|\!\neq\! 0$ , where $0 \leq r_{2} \leq n-1$, $0 \leq r_{3} \leq q-2$.
		
		For $1 \leq i \leq u$, let $\theta_{t_{j_{i}}}$ be a generator of $\mathbb{F}_{q^{k_{t_{j_{i}}}}}^{*}$.
		\!It's not hard to see that $\big\{\alpha_{t_{j_{i}}}^{q^{r_{1}}-1}\mid\alpha_{t_{j_{i}}}\!\in\! \mathbb{F}_{q^{k_{t_{j_{i}}}}}^{*}\big\}\!=\!\langle \theta_{t_{j_{i}}}^{q^{r_{1}}-1}\rangle$,
		which is a cyclic subgroup of $\mathbb{F}_{q^{m}}^{*}$ of order $I_{t_{j_{i}}}=\frac{q^{k_{t_{j_{i}}}}-1}{q^{{\rm gcd}(k_{t_{j_{i}}},r_{1})}-1}$.
		Since $\langle \xi \rangle$ is a cyclic subgroup of $\mathbb{F}_{q^{m}}^{*}$ of order $I=q-1$,
		$\langle \xi \rangle \langle \theta_{t_{j_{i}}}^{q^{r_{1}}-1} \rangle$ is a cyclic subgroup of $\mathbb{F}_{q^{m}}^{*}$ of order $\frac{II_{t_{j_{i}}}}{{\rm gcd}(I,I_{t_{j_{i}}})}$ for $1 \leq i \leq u$.
		For $1\leq i< i'\leq u$, $\langle \theta_{t_{j_{i}}}^{q^{r_{1}}-1}  \rangle \langle \theta_{t_{j_{i'}}}^{q^{r_{1}}-1} \rangle$ is a cyclic subgroup of $\mathbb{F}_{q^{m}}^{*}$ of order $\frac{I_{t_{j_{i}}}I_{t_{j_{i'}}}}{{\rm gcd}(I_{t_{j_{i}}},I_{t_{j_{i'}}})}$.
		Let
		$$S'(r_{1})=\Big\{0 \leq z \leq n-1~\Big|~\exists~0 \leq v \leq q-2 ~s.t.~ \zeta^{-(1+ra_{t_{j_{i}}})q^{r_{1}}z}\in \xi^{v}\langle \theta_{t_{j_{i}}}^{q^{r_{1}}-1}  \rangle ~{\rm for~all~}1 \leq i \leq u\Big\}.$$
		Then we see that
		\begin{align*}
			r_{2}\in S'(r_{1})~\Rightarrow~&\zeta^{-(1+ra_{t_{j_{i}}})q^{r_{1}}r_{2}}\in \langle \xi\rangle\langle \theta_{t_{j_{i}}}^{q^{r_{1}}-1}  \rangle~{\rm for}~1\leq i\leq u,\\
			& {\rm and}~\zeta^{-r(a_{t_{j_{i'}}}-a_{t_{j_{i}}})q^{r_{1}}r_{2}}\in \langle \theta_{t_{j_{i}}}^{q^{r_{1}}-1}  \rangle\langle \theta_{t_{j_{i'}}}^{q^{r_{1}}-1} \rangle~{\rm for}~1\leq i< i'\leq u\\
			\Leftrightarrow~& \frac{n}{{\rm gcd}\big(n,\frac{(1+ra_{t_{j_{i}}})II_{t_{j_{i}}}}{r\cdot {\rm gcd}(I,I_{t_{j_{i}}})}\big)}~\Big|~ r_{2}~{\rm for}~1\leq i\leq u,\\
			&{\rm and}~\frac{n}{{\rm gcd}\big(n,\frac{(a_{t_{j_{i'}}}-a_{t_{j_{i}}})I_{t_{j_{i}}}I_{t_{j_{i'}}}}{{\rm gcd}(I_{t_{j_{i}}},I_{t_{j_{i'}}})}\big)}~\Big|~ r_{2}~{\rm for}~1\leq i< i'\leq u\\
			\Leftrightarrow~&\frac{n}{{\rm gcd}(d_{1},d_{2})}~\Big|~ r_{2},
		\end{align*}
		where $$d_{1}={\rm gcd}\Big(n,\frac{(1+ra_{t_{j_{1}}})II_{t_{j_{1}}}}{r\!\cdot\!
			{\rm gcd}(I,I_{t_{j_{1}}})},\cdots,\frac{(1+ra_{t_{j_{u}}})II_{t_{j_{u}}}}{r\!\cdot\!
			{\rm gcd}(I,I_{t_{j_{u}}})}\Big)$$
		and
		\begin{align*}d_{2}={\rm gcd}\Big(n,\frac{(a_{t_{j_{2}}}\!-\!a_{t_{j_{1}}})I_{t_{j_{1}}}I_{t_{j_{2}}}}{ {\rm gcd}(I_{t_{j_{1}}},I_{t_{j_{2}}})},
			\cdots,\frac{(a_{t_{j_{u}}}\!-\!a_{t_{j_{1}}})I_{t_{j_{1}}}I_{t_{j_{u}}}}{ {\rm gcd}(I_{t_{j_{1}}},I_{t_{j_{u}}})},
			\cdots, \frac{(a_{t_{j_{u}}}\!-\!a_{t_{j_{u-1}}})I_{t_{j_{u-1}}}I_{t_{j_{u}}}}{ {\rm gcd}(I_{t_{j_{u-1}}},I_{t_{j_{u}}})}\Big).
		\end{align*}
		Hence
		\begin{align*}
			|S'(r_{1})|\leq {\rm gcd}(d_{1},d_{2})&= {\rm gcd}\Big(n,\frac{(1+ra_{t_{j_{1}}})II_{t_{j_{1}}}}{r\!\cdot\! {\rm gcd}(I,I_{t_{j_{1}}})},\cdots,\frac{(1+ra_{t_{j_{u}}})II_{t_{j_{u}}}}{r\!\cdot\!{\rm gcd}(I,I_{t_{j_{u}}})},
			\frac{(a_{t_{j_{2}}}\!-\!a_{t_{j_{1}}})I_{t_{j_{1}}}I_{t_{j_{2}}}}{{\rm gcd}(I_{t_{j_{1}}},I_{t_{j_{2}}})},\\
			&\quad \cdots,\frac{(a_{t_{j_{u}}}\!-\!a_{t_{j_{1}}})I_{t_{j_{1}}}I_{t_{j_{u}}}}{{\rm gcd}(I_{t_{j_{1}}},I_{t_{j_{u}}})}, \cdots,\frac{(a_{t_{j_{u}}}\!-\!a_{t_{j_{u-1}}})I_{t_{j_{u-1}}}I_{t_{j_{u}}}}{{\rm gcd}(I_{t_{j_{u-1}}},I_{t_{j_{u}}})}\Big).
		\end{align*}
		Suppose $r_{2}\in S'(r_{1})$, and let
		$$R'(r_{1},r_{2})=\Big\{0 \leq z \leq q-2~\Big|~\zeta^{-(1+ra_{t_{j_{i}}})q^{r_{1}}r_{2}}\in \xi^{z}\langle \theta_{t_{j_{i}}}^{q^{r_{1}}-1}  \rangle ~{\rm for~all~}1\leq i\leq u\Big\}.$$
		Similar to the proof of Theorem \ref{t3.1}, we obtain that
		$|R'(r_{1},r_{2})|={\rm gcd}(I,I_{t_{j_{1}}},\cdots,I_{t_{j_{u}}}).$
		
		We conclude that
		\begin{small}	\begin{align*}
				&\big|\langle\mu_{q},\rho,\sigma_{\xi} \rangle\backslash \mathcal{C}_{j_{1},j_{2},\cdots, j_{u}}^{\sharp}\big|\\
				=&\frac{1}{mn(q-1)}\sum_{r_{1}=0}^{m-1}\sum_{r_{2}\in S'(r_{1})}\sum_{r_{3}\in R'(r_{1},r_{2})}\prod_{i=1}^{u}(q^{{\rm gcd}(k_{t_{j_{i}}},r_{1})}-1)\\
				=&\frac{1}{mn(q-1)}\sum_{r_{1}=0}^{m-1}|S'(r_{1})|\!\cdot\!|R'(r_{1},r_{2})|\prod_{i=1}^{u}(q^{{\rm gcd}(k_{t_{j_{i}}},r_{1})}-1)\\
				\leq &\frac{1}{mn(q-1)}\sum_{r_{1}=0}^{m-1}{\rm gcd}\Big(n,\frac{(1+ra_{t_{j_{1}}})II_{t_{j_{1}}}}{r\!\cdot\!{\rm gcd}(I,I_{t_{j_{1}}})},\cdots,\frac{(1+ra_{t_{j_{u}}})II_{t_{j_{u}}}}{r\!\cdot\!{\rm gcd}(I,I_{t_{j_{u}}})},\frac{(a_{t_{j_{2}}}\!-\!a_{t_{j_{1}}})I_{t_{j_{1}}}I_{t_{j_{2}}}}{{\rm gcd}(I_{t_{j_{1}}},I_{t_{j_{2}}})},\cdots,\\
				&\frac{(a_{t_{j_{u}}}\!-\!a_{t_{j_{1}}})I_{t_{j_{1}}}I_{t_{j_{u}}}}{{\rm gcd}(I_{t_{j_{1}}},I_{t_{j_{u}}})},\cdots,
				\frac{(a_{t_{j_{u}}}\!-\!a_{t_{j_{u-1}}})I_{t_{j_{u-1}}}I_{t_{j_{u}}}}{{\rm gcd}(I_{t_{j_{u-1}}},I_{t_{j_{u}}})}\Big)
				\!\cdot\!{\rm gcd}(I,I_{t_{j_{1}}},\cdots,I_{t_{j_{u}}})\!\cdot\!
				\prod_{i=1}^{u}(q^{{\rm gcd}(k_{t_{j_{i}}},r_{1})}-1).
			\end{align*}
		\end{small}
		The proof is then completed.
	\end{proof}
	
	Based on Lemma \ref{l3.1}, an upper bound on the number of non-zero weights of $\mathcal{C}$ is derived as follows.
	\begin{Theorem}\label{t3.2}
		Let $\mathcal{C}$ be a $\lambda$-constacyclic code of length $n$ over $\mathbb{F}_{q}$. Suppose that
		$$\mathcal{C}=\mathcal{R}^{(q)}_{n,\lambda}\varepsilon_{t_{1}}\bigoplus \mathcal{R}^{(q)}_{n,\lambda}\varepsilon_{t_{2}}\bigoplus \cdots \bigoplus \mathcal{R}^{(q)}_{n,\lambda}\varepsilon_{t_{w}},$$
		where $0\leq t_{1}< t_{2}< \cdots< t_{w}\leq s$, and the primitive idempotent $\varepsilon_{t_{j}}$ corresponds to the $q$-cyclotomic coset $\{1+ra_{t_{j}},(1+ra_{t_{j}})q,\cdots,(1+ra_{t_{j}})q^{k_{t_{j}}-1}\}$ for each $1\leq j\leq w$.
		Then the number of $\langle\mu_{q},\rho,\sigma_{\xi} \rangle$-orbits of $\mathcal{C}^{*}=\mathcal{C}\backslash \{\bf 0\}$ is less than or equal to
		$$\sum_{\substack{\{j_{1},j_{2},\cdots,j_{u}\}\subseteq \{1,2,\cdots,w\}\\ 1\leq j_{1}< j_{2}<\cdots<j_{u}\leq w}}N_{j_{1},j_{2},\cdots, j_{u}},$$
		where the value of $N_{j_{1},j_{2},\cdots, j_{u}}$ is
		\begin{align*}
			&\frac{1}{mn(q-1)}\sum_{h=0}^{m-1}
			{\rm gcd}\Big(n,\frac{(1+ra_{t_{j_{1}}})II_{t_{j_{1}}}}{r\!\cdot\!{\rm gcd}(I,I_{t_{j_{1}}})},
			\cdots,\frac{(1+ra_{t_{j_{u}}})II_{t_{j_{u}}}}{r\!\cdot\!{\rm gcd}(I,I_{t_{j_{u}}})},
			\frac{(a_{t_{j_{2}}}\!-\!a_{t_{j_{1}}})I_{t_{j_{1}}}I_{t_{j_{2}}}}{{\rm gcd}(I_{t_{j_{1}}},I_{t_{j_{2}}})},\cdots,\\
			&\frac{(a_{t_{j_{u}}}\!-\!a_{t_{j_{1}}})I_{t_{j_{1}}}I_{t_{j_{u}}}}
			{{\rm gcd}(I_{t_{j_{1}}},I_{t_{j_{u}}})},\cdots,
			\frac{(a_{t_{j_{u}}}\!-\!a_{t_{j_{u-1}}})I_{t_{j_{u-1}}}I_{t_{j_{u}}}}
			{{\rm gcd}(I_{t_{j_{u-1}}},I_{t_{j_{u}}})}\Big)\!\cdot\! {\rm gcd}(I,I_{t_{j_{1}}},\cdots,I_{t_{j_{u}}})\!\cdot\!\prod_{i=1}^{u}(q^{{\rm gcd}(k_{t_{j_{i}}},h)}\!-\!1)
		\end{align*}
		with $I=q-1$ and $I_{t_{j_{i}}}=\frac{q^{k_{t_{j_{i}}}}-1}{q^{{\rm gcd}(k_{t_{j_{i}}},h)}-1}$ for $1 \leq i \leq u$.
		
		In particular, the number of non-zero weights of $\mathcal{C}$ is less than or equal to the number of $\langle\mu_{q},\rho,\sigma_{\xi} \rangle$-orbits of $\mathcal{C}^{*}$.
	\end{Theorem}
	
	\begin{proof}
		Note that
		$$\mathcal{C}^{*}=\bigcup_{\substack{\{j_{1},j_{2},\cdots,j_{u}\}\subseteq \{1,2,\cdots,w\}\\ 1\leq j_{1}< j_{2}<\cdots<j_{u}\leq w}}\mathcal{C}_{j_{1},j_{2},\cdots, j_{u}}^{\sharp}$$
		is a disjoint union. Thus
		$$\big|\langle\mu_{q},\rho,\sigma_{\xi} \rangle\backslash \mathcal{C}^{*}\big|=\sum_{\substack{\{j_{1},j_{2},\cdots,j_{u}\}\subseteq \{1,2,\cdots,w\}\\ 1\leq j_{1}< j_{2}<\cdots<j_{u}\leq w}}
		\big|\langle\mu_{q},\rho,\sigma_{\xi} \rangle\backslash \mathcal{C}_{j_{1},j_{2},\cdots, j_{u}}^{\sharp}\big|.$$
		The desired result then follows from Lemma \ref{l3.1}.	
	\end{proof}
	
	\begin{Remark}\label{Re.1}{\rm
			(1) Let $\mathcal{C}$ be the $\lambda$-constacyclic code in Theorem \ref{t3.2}. In \cite[Lemma 3.6]{21}, the authors assumed ${\rm gcd}\big( \frac{q-1}{r}, r \big) = 1$ and then presented that
			\begin{align*}
				\big|\langle\rho,\sigma_{\xi} \rangle\backslash \mathcal{C}^{*}\big|=&\sum_{\substack{\{j_{1},j_{2},\cdots,j_{u}\}\subseteq \{1,2,\cdots,w\}\\ 1\leq j_{1}< j_{2}<\cdots<j_{u}\leq w}}\frac{1}{rn(q-1)}{\rm gcd}(n,1+ra_{t_{j_1}},\cdots,1+ra_{t_{j_u}})\\
				&\cdot {\rm gcd}\Big(q-1,\frac{rn}{{\rm gcd}(n,1+ra_{t_{j_1}})},
				\cdots,\frac{rn}{{\rm gcd}(n,1+ra_{t_{j_u}})}\Big)
				\!\cdot\!\prod\limits_{i=1}^{u}(q^{k_{t_{j_{i}}}}-1).
			\end{align*}
			The right-hand side of the equation above is in fact an upper bound of $\big|\langle\rho,\sigma_{\xi} \rangle\backslash \mathcal{C}^{*}\big|$ rather than the value of $\big|\langle\rho,\sigma_{\xi} \rangle\backslash \mathcal{C}^{*}\big|$. This is because in the proof of Theorem \ref{t3.2} the following two conditions were considered to be equivalent:\\
			\quad 1) $\rho^{z}({\bf c}_{t_{j_{i}}})=\xi^{h}{\bf c}_{t_{j_{i}}}$ for all $i=1,2,\cdots,u$,\\
			\quad 2) $\frac{q-1}{{\rm gcd}(rh,q-1)}\big|\frac{rn}{{\rm gcd}(n,1+r{a}_{t_{j_i}})}$ for all $i=1,2,\cdots,u$, where $1 \leq h \leq \frac{q-1}{r}$.\\
			However, condition 2) does not imply condition 1) as condition 1) is also dependent on the choice of integer $z$.

			In the proof of Lemma \ref{l3.1}, let $r_1=0$, then we can obtain
			\begin{small}
				\begin{equation}\label{E.2}
					\big|\langle\rho,\sigma_{\xi} \rangle\backslash \mathcal{C}^{*}\big|
					\!=\! \sum_{\substack{\{j_{1},j_{2},\cdots,j_{u}\}\subseteq \{1,2,\cdots,w\}\\ 1\leq j_{1}< j_{2}<\cdots<j_{u}\leq w}}
					\frac{\prod\limits_{i=1}^{u}(q^{k_{t_{j_{i}}}}\!-\!1)}{rn(q\!-\!1)}
					{\rm gcd}\big(rn,\!(1\!+\!ra_{t_{j_{1}}})(q\!-\!1),\!
					r(a_{t_{j_{2}}}\!-\!a_{t_{j_{1}}}),\!\cdots,\!r(a_{t_{j_{u}}}\!-\!a_{t_{j_{1}}})\big).
				\end{equation}
			\end{small}
			Furthermore, the above equation does not require ${\rm gcd}\big( \frac{q-1}{r}, r \big) = 1$.
			
			(2) We assert that the upper bound on the number of non-zero weights of $\mathcal{C}$ given by Theorem \ref{t3.2} is better than the upper bound $\big|\langle\rho,\sigma_{\xi} \rangle\backslash \mathcal{C}^{*}\big|$, and of course better than the upper bound given by  \cite[Lemma 3.6]{21}. To prove this, it suffices to show that
			$$N_{j_{1},j_{2},\cdots, j_{u}}\leq \frac{1}{rn(q-1)}{\rm gcd}\Big(rn,(1+a_{t_{j_{1}}})(q-1),
			r(a_{t_{j_{2}}}\!-\!a_{t_{j_{1}}}),
			\cdots,r(a_{t_{j_{u}}}\!-\!a_{t_{j_{1}}})\Big)\cdot\prod\limits_{i=1}^{u}(q^{k_{t_{j_{i}}}}-1).$$
			Indeed,
			\begin{small}
				\begin{align*}
					N_{j_{1},j_{2},\cdots, j_{u}}
					&\leq \frac{1}{mn(q-1)}\sum_{r_{1}=0}^{m-1}{\rm gcd}
					\Big(n,\frac{(1+ra_{t_{j_{1}}})II_{t_{j_{1}}}}{r\!\cdot\!{\rm gcd}(I,I_{t_{j_{1}}})},\cdots,
					\frac{(1+ra_{t_{j_{u}}})II_{t_{j_{u}}}}{r\!\cdot\!{\rm gcd}(I,I_{t_{j_{u}}})},
					\frac{(a_{t_{j_{2}}}\!-\!a_{t_{j_{1}}})I_{t_{j_{1}}}I_{t_{j_{2}}}}{{\rm gcd}(I_{t_{j_{1}}},I_{t_{j_{2}}})},\cdots, \\
					&\quad\frac{(a_{t_{j_{u}}}\!-\!a_{t_{j_{1}}})I_{t_{j_{1}}}I_{t_{j_{u}}}}{{\rm gcd}(I_{t_{j_{1}}},I_{t_{j_{u}}})},\!\cdots\!,
					\frac{(a_{t_{j_{u}}}\!-\!a_{t_{j_{u-1}}})I_{t_{j_{u-1}}}I_{t_{j_{u}}}}{{\rm gcd}(I_{t_{j_{u-1}}},I_{t_{j_{u}}})}\Big)
					\!\cdot\!{\rm gcd}(I\!,I_{t_{j_{1}}}\!,\cdots\!,I_{t_{j_{u}}})
					\cdot\! \prod_{i=1}^{u}(q^{{\rm gcd}(k_{t_{j_{i}}},r_{1})}\!-\!1)\\
					&\leq \frac{1}{rmn(q-1)}\sum_{r_{1}=0}^{m-1}
					{\rm gcd}\Big(rn,(1+ra_{t_{j_{1}}})II_{t_{j_{1}}},\cdots,
					(1+ra_{t_{j_{u}}})II_{t_{j_{u}}},r(a_{t_{j_{2}}}\!-\!a_{t_{j_{1}}})I_{t_{j_{1}}}I_{t_{j_{2}}},\cdots,\\
					&\quad r(a_{t_{j_{u}}}\!-\!a_{t_{j_{1}}})I_{t_{j_{1}}}I_{t_{j_{u}}},
					\cdots, r(a_{t_{j_{u}}}\!-\!a_{t_{j_{u-1}}})
					I_{t_{j_{u-1}}}I_{t_{j_{u}}}\Big)\cdot \prod_{i=1}^{u}(q^{{\rm gcd}(k_{t_{j_{i}}},r_{1})}-1)\\
					&\leq \frac{1}{rmn(q-1)}\sum_{r_{1}=0}^{m-1}{\rm gcd}
					\Big(rn,(1+ra_{t_{j_{1}}})I,\cdots,(1+ra_{t_{j_{u}}})I,
					r(a_{t_{j_{2}}}\!-\!a_{t_{j_{1}}}),\cdots, \\
					&\quad r(a_{t_{j_{u}}}\!-\!a_{t_{j_{1}}}),\cdots,r(a_{t_{j_{u}}}\!-\!a_{t_{j_{u-1}}})\Big)
					\!\cdot\! I_{t_{j_{1}}} I_{t_{j_{2}}}\cdots I_{t_{j_{u}}}
					\!\cdot\! \prod_{i=1}^{u}(q^{{\rm gcd}(k_{t_{j_{i}}},r_{1})}-1)\\
					&=\frac{1}{rmn(q-1)}\sum_{r_{1}=0}^{m-1}\!{\rm gcd}
					\Big(rn,(1+ra_{t_{j_{1}}})(q-1),\!\cdots\!,\!(1+ra_{t_{j_{u}}})(q-1),\!
					r(a_{t_{j_{2}}}\!-\!a_{t_{j_{1}}}),\\
					&\quad \cdots, r(a_{t_{j_{u}}}\!-\!a_{t_{j_{1}}}),\cdots,r(a_{t_{j_{u}}}\!-\!a_{t_{j_{u-1}}})\Big)
					\cdot \prod\limits_{i=1}^{u}(q^{k_{t_{j_{i}}}}-1)\\
					&=\frac{1}{rn(q-1)}
					{\rm gcd}\Big(rn,(1+ra_{t_{j_{1}}})(q-1),\cdots,(1+ra_{t_{j_{u}}})(q-1),
					r(a_{t_{j_{2}}}\!-\!a_{t_{j_{1}}}),\\
					&\quad \cdots,r(a_{t_{j_{u}}}\!-\!a_{t_{j_{1}}}),\cdots, r(a_{t_{j_{u}}}\!-\!a_{t_{j_{u-1}}})\Big)
					\cdot \prod\limits_{i=1}^{u}(q^{k_{t_{j_{i}}}}-1)\\
					&=\frac{1}{rn(q-1)}
					{\rm gcd}\Big(rn,(1+ra_{t_{j_{1}}})(q-1), r(a_{t_{j_{2}}}\!-\!a_{t_{j_{1}}}),\cdots,r(a_{t_{j_{u}}}\!-\!a_{t_{j_{1}}})\Big)
					\cdot \prod\limits_{i=1}^{u}(q^{k_{t_{j_{i}}}}-1).
				\end{align*}
			\end{small}
			This completes the proof.}
	\end{Remark}
	
	We discussed above the action of the group $\langle\mu_{q},\rho,\sigma_{\xi} \rangle$ on a general $\lambda$-constacyclic code $\mathcal{C}$, giving an upper bound on the number of $\langle\mu_{q},\rho,\sigma_{\xi} \rangle$-orbits of $\mathcal{C}^{*}=\mathcal{C}\backslash \{\bf 0\}$.
	Now let's look at some special cases. In these cases, we can explicitly give a formula for the number of $\langle\mu_{q},\rho,\sigma_{\xi} \rangle$-orbits of $\mathcal{C}^{*}$.
	
	\begin{Theorem}\label{t3.3}
		Let 		$\mathcal{C}=\mathcal{R}^{(q)}_{n,\lambda}\varepsilon_{t_{1}}\bigoplus\mathcal{R}^{(q)}_{n,\lambda}\varepsilon_{t_{2}},$
		where $0\leq t_{1} < t_{2}\leq s$,
		and the primitive idempotent $\varepsilon_{t_{i}}$ corresponds to the $q$-cyclotomic coset $\{1+ra_{t_{i}},(1+ra_{t_{i}})q,\cdots,(1+ra_{t_{i}})q^{k_{t_{i}}-1}\}$ for $i=1,2$.
		If $k_{t_{1}}\,|\,k_{t_{2}}$, then the number of $\langle\mu_{q},\rho,\sigma_{\xi} \rangle$-orbits of $\mathcal{C}^{*}=\mathcal{C}\backslash \{\bf 0\}$ is equal to
		$s_{t_{1}}+s_{t_{2}}+s_{t_{1},t_{2}},$
		where
		$$s_{t_{i}}=\frac{1}{k_{t_{i}}}\sum_{h\mid k_{t_{i}}}\varphi(\frac{k_{t_{i}}}{h}){\rm gcd}\big(q^{h}-1,\frac{q^{k_{t_{i}}}-1}{q-1},\frac{(1+ra_{t_{i}})(q^{k_{t_{i}}}-1)}{rn}\big)~~{\rm for}~i=1,2,$$
		and
		\begin{small}
			\begin{align*}
				s_{t_{1},t_{2}}&=\frac{1}{m}\sum_{h=0}^{m-1}{\rm gcd}\Big((q^{{\rm gcd}(k_{t_{1}},h)}-1)\cdot {\rm gcd}\big(q^{{\rm gcd}(k_{t_{2}},h)}-1,\frac{(q^{k_{t_{1}}}-1)(q^{{\rm gcd}(k_{t_{2}},h)}-1)}{(q-1)(q^{{\rm gcd}(k_{t_{1}},h)}-1)},\frac{(1+ra_{t_{2}})(q^{k_{t_{2}}}-1)}{rn},\\
				&\quad~\frac{(1+ra_{t_{1}})(q^{k_{t_{1}}}-1)(q^{{\rm gcd}(k_{t_{2}},h)}-1)}{rn(q^{{\rm gcd}(k_{t_{1}},h)}-1)}\big), \frac{(a_{t_{2}}-a_{t_{1}})(q^{k_{t_{1}}}-1)(q^{k_{t_{2}}}-1)}{n(q-1)}\Big).
			\end{align*}
		\end{small}
		
		\noindent In particular, the number of non-zero weights of $\mathcal{C}$ is less than or equal to the number of $\langle\mu_{q},\rho,\sigma_{\xi} \rangle$-orbits of $\mathcal{C}^{*}$, with equality if and only if for any two codewords ${\bf c}_{1},{\bf c}_{2}\in \mathcal{C}^{*}$ with the same weight, there exist integers $j_{1}$, $j_{2}$ and $ j_{3}$ such that $\mu_{q}^{j_{1}}\rho^{j_{2}}(\xi^{j_{3}}{\bf c}_{1})={\bf c}_{2}$,
		where $0 \leq j_{1} \leq m-1$, $0 \leq j_{2} \leq n-1$ and $0 \leq j_{3} \leq q-2$.
	\end{Theorem}
	\begin{proof}
		Let $s_{t_{1}}$, $s_{t_{2}}$ and $s_{t_{1},t_{2}}$ denote the number of $\langle\mu_{q},\rho,\sigma_{\xi} \rangle$-orbits of $\mathcal{R}^{(q)}_{n,\lambda}\varepsilon_{t_{1}}\backslash \{\bf 0\}$, $\mathcal{R}^{(q)}_{n,\lambda}\varepsilon_{t_{2}}\backslash \{\bf 0\}$ and $\mathcal{R}^{(q)}_{n,\lambda}\varepsilon_{t_{1}}\backslash \{{\bf 0}\}\bigoplus\mathcal{R}^{(q)}_{n,\lambda}\varepsilon_{t_{2}}\backslash \{{\bf 0}\}$, respectively.
		Then $\big|\langle\mu_{q},\rho,\sigma_{\xi} \rangle\backslash \mathcal{C}^{*}\big|=s_{t_{1}}+s_{t_{2}}+s_{t_{1},t_{2}}.$
		It follows from Theorem \ref{t3.1} that
		$$s_{t_{i}}=\frac{1}{k_{t_{i}}}\sum_{h\mid k_{t_{i}}}\varphi(\frac{k_{t_{i}}}{h}){\rm gcd}\big(q^{h}-1,\frac{q^{k_{t_{i}}}-1}{q-1},\frac{(1+ra_{t_{i}})(q^{k_{t_{i}}}-1)}{rn}\big)~~{\rm for}~i=1,2.$$
		In the proof of Theorem \ref{t3.2}, let $u=2$, $j_1=1$ and $j_2=2$.
		Since $k_{t_{1}}\,|\,k_{t_{2}}$, we have $(q^{k_{t_{1}}}-1)\,|\,(q^{k_{t_{2}}}-1)$, then $\mathbb{F}_{q^{k_{t_{1}}}}^{*}\leq \mathbb{F}_{q^{k_{t_{2}}}}^{*}$, that is, $\mathbb{F}_{q^{k_{t_{1}}}}^{*}$ is a subgroup of $\mathbb{F}_{q^{k_{t_{2}}}}^{*}$.
		Let $\mathbb{F}_{q^{k_{t_{1}}}}^{*}=\langle  \theta_{t_{1}} \rangle$ and $\mathbb{F}_{q^{k_{t_{2}}}}^{*}=\langle  \theta_{t_{2}} \rangle$, then $\theta_{t_{1}}=\theta_{t_{2}}^{i}$ for some non-negative integer $i$.
		Thus $\langle \theta_{t_{1}}^{q^{r_{1}}-1}  \rangle=\langle \theta_{t_{2}}^{i(q^{r_{1}}-1)}  \rangle \leq \langle \theta_{t_{2}}^{q^{r_{1}}-1}  \rangle$.
		Then one can easily check that
		\begin{align*}
			r_{2}\in S'(r_{1})
			~&\Leftrightarrow~\zeta^{-(1+ra_{t_{i}})q^{r_{1}}r_{2}}\in \langle\xi\rangle\langle \theta_{t_{i}}^{q^{r_{1}}-1}\rangle~{\rm for}~i=1,2 ~{\rm and}~\zeta^{-r(a_{t_{2}}-a_{t_{1}})q^{r_{1}}r_{2}}\in \langle \theta_{t_{2}}^{q^{r_{1}}-1}  \rangle\\
			&\Leftrightarrow~\zeta^{-(1+ra_{t_{1}})q^{r_{1}}r_{2}}\in \langle\xi\rangle\langle \theta_{t_{1}}^{q^{r_{1}}-1}\rangle~{\rm and}~\zeta^{-r(a_{t_{2}}-a_{t_{1}})q^{r_{1}}r_{2}}\in \langle \theta_{t_{2}}^{q^{r_{1}}-1}  \rangle.
		\end{align*}
		We see from the proof of Lemma \ref{l3.1} that
		\begin{small}
			\begin{align*}
				s_{t_{1},t_{2}}
				=&\frac{1}{mn(q-1)}\sum_{r_{1}=0}^{m-1}|S'(r_{1})|\!\cdot\!|R'(r_{1},r_{2})|
				\cdot\!\prod_{i=1}^{2}(q^{{\rm gcd}(k_{t_{i}},r_{1})}-1)\\
				=&\frac{1}{mn(q-1)}\sum_{r_{1}=0}^{m-1}{\rm gcd}\Big(n,
				(a_{t_{2}}-a_{t_{1}})|\langle \theta_{t_{2}}^{q^{r_{1}}-1}\rangle|,
				\frac{(1+ra_{t_{1}})(q-1)|\langle \theta_{t_{1}}^{q^{r_{1}}-1}\rangle|}{r\!\cdot\!{\rm gcd}(q-1,|\langle \theta_{t_{1}}^{q^{r_{1}}-1}\rangle|)}\Big)\cdot{\rm gcd}\big(q-1,|\langle \theta_{t_{1}}^{q^{r_{1}}-1}\rangle|\big)\\
				&\cdot\prod_{i=1}^{2}\big(q^{{\rm gcd}(k_{t_{i}},r_{1})}-1\big)\\
				=&\frac{1}{mn(q-1)}\sum_{r_{1}=0}^{m-1}
				{\rm gcd}\Big(n(q-1),n|\langle\theta_{t_{1}}^{q^{r_{1}}-1}\rangle|,
				(a_{t_{2}}-a_{t_{1}})(q-1)|\langle \theta_{t_{2}}^{q^{r_{1}}-1}\rangle|,
				\frac{(1+ra_{t_{1}})(q-1)|\langle \theta_{t_{1}}^{q^{r_{1}}-1}\rangle|}{r},\\
				&(a_{t_{2}}-a_{t_{1}})|\langle \theta_{t_{1}}^{q^{r_{1}}-1}\rangle||\langle \theta_{t_{2}}^{q^{r_{1}}-1}\rangle|\Big) \cdot\prod_{i=1}^{2}\big(q^{{\rm gcd}(k_{t_{i}},r_{1})}-1\big)\\
				=&\frac{1}{mn(q-1)}\sum_{r_{1}=0}^{m-1}{\rm gcd}\Big(n(q-1),n|\langle \theta_{t_{1}}^{q^{r_{1}}-1}\rangle|,\frac{(1+ra_{t_{1}})(q-1)|\langle \theta_{t_{1}}^{q^{r_{1}}-1}\rangle|}{r},\frac{(1+ra_{t_{2}})(q-1)|\langle \theta_{t_{2}}^{q^{r_{1}}-1}\rangle|}{r},\\
				&(a_{t_{2}}-a_{t_{1}})|\langle \theta_{t_{1}}^{q^{r_{1}}-1}\rangle||\langle \theta_{t_{2}}^{q^{r_{1}}-1}\rangle|\Big)
				\cdot \prod_{i=1}^{2}\big(q^{{\rm gcd}(k_{t_{i}},r_{1})}-1\big)\\
				=&\frac{1}{m}\sum_{r_{1}=0}^{m-1}{\rm gcd}\Big(\prod_{i=1}^{2}(q^{{\rm gcd}(k_{t_{i}},r_{1})}-1),\frac{(q^{k_{t_{1}}}-1)(q^{{\rm gcd}(k_{t_{2}},r_{1})}-1)}{q-1},\frac{(1+ra_{t_{1}})(q^{k_{t_{1}}}-1)(q^{{\rm gcd}(k_{t_{2}},r_{1})}-1)}{rn},\\
				&\frac{(1+ra_{t_{2}})(q^{k_{t_{2}}}-1)(q^{{\rm gcd}(k_{t_{1}},r_{1})}-1)}{rn}, \frac{(a_{t_{2}}-a_{t_{1}})(q^{k_{t_{1}}}-1)(q^{k_{t_{2}}}-1)}{n(q-1)}\Big)\\
				=&\frac{1}{m}\sum_{r_{1}=0}^{m-1}{\rm gcd}\Big((q^{{\rm gcd}(k_{t_{1}},r_{1})}-1)\!\cdot\!{\rm gcd}\big(q^{{\rm gcd}(k_{t_{2}},r_{1})}-1,\frac{(q^{k_{t_{1}}}-1)(q^{{\rm gcd}(k_{t_{2}},r_{1})}-1)}{(q-1)(q^{{\rm gcd}(k_{t_{1}},r_{1})}-1)},\frac{(1+ra_{t_{2}})(q^{k_{t_{2}}}-1)}{rn},\\
				&\frac{(1+ra_{t_{1}})(q^{k_{t_{1}}}-1)(q^{{\rm gcd}(k_{t_{2}},r_{1})}-1)}{rn(q^{{\rm gcd}(k_{t_{1}},r_{1})}-1)}\big), \frac{(a_{t_{2}}-a_{t_{1}})(q^{k_{t_{1}}}-1)(q^{k_{t_{2}}}-1)}{n(q-1)}\Big).
			\end{align*}
		\end{small}
		
		\noindent We are done.
	\end{proof}

	The following two conclusions are direct corollaries of Theorem \ref{t3.3}.
	
	\begin{Corollary}\label{c3.3}
		Let $\mathcal{C}=\mathcal{R}^{(q)}_{n,\lambda}\varepsilon_{t_{1}}\bigoplus\mathcal{R}^{(q)}_{n,\lambda}\varepsilon_{t_{2}},$ where $0\leq t_{1} < t_{2}\leq s$, and the primitive idempotent $\varepsilon_{t_{i}}$ corresponds to the $q$-cyclotomic coset $\{1+ra_{t_{i}},(1+ra_{t_{i}})q,\cdots,(1+ra_{t_{i}})q^{k_{t_{i}}-1}\}$ for $i=1,2$. If $k_{t_{1}}=1$ and $k_{t_{2}}=k$, then the number of $\langle\mu_{q},\rho,\sigma_{\xi} \rangle$-orbits of $\mathcal{C}^{*}=\mathcal{C}\backslash \{\bf 0\}$ is equal to
		\begin{small}
			$$1+\frac{1}{k}\sum_{h\mid k}\varphi\big(\frac{k}{h}\big)\left(
			{\rm gcd}\Big(q^{h}-1,\frac{q^{k}-1}{q-1},\frac{(1+ra_{t_{2}})(q^{k}-1)}{rn}\Big)
			+{\rm gcd}\Big(q^{h}-1,\frac{(a_{t_2}-a_{t_1})(q^{k}-1)}{n}\Big)\right).$$
		\end{small}
		
		\noindent In particular, the number of non-zero weights of $\mathcal{C}$ is less than or equal to the number of $\langle\mu_{q},\rho,\sigma_{\xi} \rangle$-orbits of $\mathcal{C}^{*}$, with equality if and only if for any two codewords ${\bf c}_{1},{\bf c}_{2}\in \mathcal{C}^{*}$ with the same weight, there exist integers $j_{1}$, $j_{2}$ and $ j_{3}$ such that $\mu_{q}^{j_{1}}\rho^{j_{2}}(\xi^{j_{3}}{\bf c}_{1})={\bf c}_{2}$,
		where $0 \leq j_{1} \leq m-1$, $0 \leq j_{2} \leq n-1$ and $0 \leq j_{3} \leq q-2$.
		
	\end{Corollary}
	
	\begin{proof}
		According to Theorem \ref{t3.3} and its proof, we have
		$\big|\langle\mu_{q},\rho,\sigma_{\xi} \rangle\backslash \mathcal{C}^{*}\big|=s_{t_{1}}+s_{t_{2}}+s_{t_{1},t_{2}},$
		where
		$$
		s_{t_{1}}={\rm gcd}\Big(q-1,1,\frac{(1+ra_{t_{1}})(q-1)}{rn}\Big)=1,\quad
		s_{t_{2}}=\frac{1}{k}\sum_{h\mid k}\varphi\big(\frac{k}{h}\big){\rm gcd}\Big(q^{h}-1,\frac{q^{k}-1}{q-1},\frac{(1+ra_{t_{2}})(q^{k}-1)}{rn}\Big)
		$$
		and
		\begin{small}
			\begin{align*}
				s_{t_{1},t_{2}}
				=&\frac{1}{m}\sum_{r_{1}=0}^{m-1}{\rm gcd}\Big((q\!-\!1)\!\cdot\!{\rm gcd}\big(q^{{\rm gcd}(k,r_{1})}\!-\!1,\frac{(q^{{\rm gcd}(k,r_{1})}\!-\!1)}{q\!-\!1},\frac{(1\!+\!ra_{t_{1}})(q^{{\rm gcd}(k,r_{1})}\!-\!1)}{rn},\frac{(1\!+\!ra_{t_{2}})(q^{k}\!-\!1)}{rn}\big),\\
				&\frac{(a_{t_2}\!-\!a_{t_1})(q^{k}\!-\!1)}{n}\Big)\\
				=&\frac{1}{m}\sum_{r_{1}=0}^{m-1}{\rm gcd}\Big(q^{{\rm gcd}(k,r_{1})}\!-\!1,\frac{(1\!+\!ra_{t_{1}})(q\!-\!1)(q^{{\rm gcd}(k,r_{1})}\!-\!1)}{rn},\frac{(1\!+\!ra_{t_{2}})(q\!-\!1)(q^{k}\!-\!1)}{rn},
				\frac{(a_{t_2}\!-\!a_{t_1})(q^{k}\!-\!1)}{n}\Big)\\
				=&\frac{1}{m}\sum_{r_{1}=0}^{m-1}{\rm gcd}\Big(q^{{\rm gcd}(k,r_{1})}\!-1,\frac{(1\!+\!ra_{t_{1}})(q\!-\!1)(q^{{\rm gcd}(k,r_{1})}\!-\!1)}{rn},\frac{(a_{t_2}\!-\!a_{t_1})(q\!-\!1)(q^{k}\!-\!1)}{n},
				\frac{(a_{t_2}\!-\!a_{t_1})(q^{k}\!-\!1)}{n}\Big)\\
				=&\frac{1}{m}\sum_{r_{1}=0}^{m-1}{\rm gcd}\Big(q^{{\rm gcd}(k,r_{1})}-1,\frac{(a_{t_2}-a_{t_1})(q^{k}-1)}{n}\Big)\\
				=&\frac{1}{k}\sum_{h|k}\varphi\big(\frac{k}{h}\big){\rm gcd}\Big(q^{h}-1,\frac{(a_{t_2}-a_{t_1})(q^{k}-1)}{n}\Big).
			\end{align*}
		\end{small}
		We have completed the proof of the corrollary.
	\end{proof}

	We present an example to illustrate that the upper bound in Corollary \ref{c3.3} improves the upper bound $\big|\langle\rho,\sigma_{\xi} \rangle\backslash \mathcal{C}^{*}\big|$ as stated in Equation (\ref{E.2}).
	
	\begin{Example}\label{e3.4}{\rm
			Take $q=5$, $n=39$ and $\lambda=-1$.
			Let $\ell$ be the number of non-zero weights of the negacyclic code
			$\mathcal{C}=\mathcal{R}_{n,\lambda}^{(q)}\varepsilon_{0}\bigoplus\mathcal{R}_{n,\lambda}^{(q)}\varepsilon_{9}$,
			where the primitive idempotents $\varepsilon_{0}$ and $\varepsilon_{9}$ correspond to the $5$-cyclotomic 			cosets $\Gamma_{0}=\{1, 5, 25, 47\}$ and $\Gamma_{9}=\{39\}$, respectively.
			\!By Equation (\ref{E.2}), we have $\ell\!\leq\!21$. By Corollary \ref{c3.3}, we have $\ell\!\leq\!13$.
			After using Magma \cite{4}, the weight enumerator of $\mathcal{C}$ is
			$1+156x^{25}+468x^{28}+156x^{30}+780x^{31}+780x^{32}+312x^{33}+156x^{34}+312x^{35}+4x^{39}$,
			which implies that the exact value of $\ell$ is 9.}
	\end{Example}
	
	\begin{Corollary}\label{c3.4}
		Let $\mathcal{C}=\mathcal{R}^{(q)}_{n,\lambda}\varepsilon_{t_{1}}\bigoplus\mathcal{R}^{(q)}_{n,\lambda}\varepsilon_{t_{2}},$
		where $0\leq t_{1} < t_{2}\leq s$, and the primitive idempotent $\varepsilon_{t_{i}}$ corresponds to the $q$-cyclotomic coset $\{1+ra_{t_{i}},(1+ra_{t_{i}})q,\cdots,(1+ra_{t_{i}})q^{k_{t_{i}}-1}\}$ for $i=1,2$. If $k_{t_{1}}=k_{t_{2}}=k$, then the number of $\langle\mu_{q},\rho,\sigma_{\xi} \rangle$-orbits of $\mathcal{C}^{*}=\mathcal{C}\backslash \{\bf 0\}$ is equal to
		\begin{small}
			\begin{align*}
				&\frac{1}{k}\sum_{h\mid k}\varphi\big(\frac{k}{h}\big)\left(\!
				{\rm gcd}\big(q^{h}-1,\frac{q^{k}-1}{q-1},\frac{(1+ra_{t_{1}})(q^{k}-1)}{rn}\big)+
				{\rm gcd}\big(q^{h}-1,\frac{q^{k}-1}{q-1},\frac{(1+ra_{t_{2}})(q^{k}-1)}{rn}\big)\right. \\
				&\left.\!+{\rm gcd}\big(\!(q^{h}-1)\!\cdot\!{\rm gcd}\big(q^{h}-1,\!\frac{q^{k}-1}{q-1},\!\frac{(1+ra_{t_{1}})(q^{k}-1)}{rn},\!\frac{(1+ra_{t_{2}})(q^{k}-1)}{rn}\big),\!
				\frac{(a_{t_2}-a_{t_1})(q^{k}-1)^{2}}{n(q-1)}\!\big)\!\right).
			\end{align*}
		\end{small}
		
		\noindent In particular, the number of non-zero weights of $\mathcal{C}$ is less than or equal to the number of $\langle\mu_{q},\rho,\sigma_{\xi} \rangle$-orbits of $\mathcal{C}^{*}$, with equality if and only if for any two codewords ${\bf c}_{1},{\bf c}_{2}\in \mathcal{C}^{*}$ with the same weight, there exist integers $j_{1}$, $j_{2}$ and $ j_{3}$ such that $\mu_{q}^{j_{1}}\rho^{j_{2}}(\xi^{j_{3}}{\bf c}_{1})={\bf c}_{2}$,
		where $0 \leq j_{1} \leq m-1$, $0 \leq j_{2} \leq n-1$ and $0 \leq j_{3} \leq q-2$.
	\end{Corollary}
	\begin{proof}
		According to Theorem \ref{t3.3} and its proof, we have
		$\big|\langle\mu_{q},\rho,\sigma_{\xi} \rangle\backslash \mathcal{C}^{*}\big|=s_{t_{1}}+s_{t_{2}}+s_{t_{1},t_{2}},$
		where
		$$s_{t_{i}}=\frac{1}{k}\sum_{h\mid k}\varphi\big(\frac{k}{h}\big){\rm gcd}\big(q^{h}-1,\frac{q^{k}-1}{q-1},\frac{(1+ra_{t_{i}})(q^{k}-1)}{rn}\big)~~{\rm for}~i=1,2,$$
		and
		\begin{small}
			\begin{align*}
				s_{t_{1},t_{2}}&=\frac{1}{m}\sum_{r_{1}=0}^{m-1}
				{\rm gcd}\Big((q^{{\rm gcd}(k,r_{1})}-1)
				\!\cdot\!
				{\rm gcd}\big(q^{{\rm gcd}(k,r_{1})}-1,\frac{q^{k}-1}{q-1},
				\frac{(1+ra_{t_{1}})(q^{k}-1)}{rn},\frac{(1+ra_{t_{2}})(q^{k}-1)}{rn}\big), \\
				&\quad~\frac{(a_{t_2}-a_{t_1})(q^{k}-1)^{2}}{n(q-1)}\Big)\\
				&=\frac{1}{k}\sum_{h|k}\varphi\big(\frac{k}{h}\big){\rm gcd}\Big((q^{h}\!-\!1)
				\!\cdot\!{\rm gcd}\big(q^{h}\!-\!1,\frac{q^{k}\!-\!1}{q\!-\!1},
				\frac{(1\!+\!ra_{t_{1}})(q^{k}\!-\!1)}{rn},\frac{(1\!+\!ra_{t_{2}})(q^{k}\!-\!1)}{rn}\big), \frac{(a_{t_2}-a_{t_1})(q^{k}\!-\!1)^{2}}{n(q\!-\!1)}\Big).
			\end{align*}
		\end{small}
		The proof is then completed.
	\end{proof}
	
	\begin{Example}{\rm
			Take $q=3$, $n=20$ and $\lambda=-1$.
			Let $\ell$ be the number of non-zero weights of the negacyclic code $\mathcal{C}=\mathcal{R}^{(q)}_{n,\lambda}\varepsilon_{1}\bigoplus\mathcal{R}^{(q)}_{n,\lambda}\varepsilon_{5}$,
			where the primitive idempotents $\varepsilon_{1}$ and $\varepsilon_{5}$ correspond to the $3$-cyclotomic cosets $\Gamma_{1}=\{5, 15 \}$ and $\Gamma_{5}=\{25, 35\}$, respectively.
			By Equation (\ref{E.2}), we have $\ell\leq 10$. Using Corollary \ref{c3.4}, we have $\ell\leq 7.$ After using Magma \cite{4}, we know that the weight enumerator of $\mathcal{C}$ is $1+8x^{5}+24x^{10}+32x^{15}+16x^{20}$, which implies that $\ell=4$.}
	\end{Example}

	\subsection{New upper bounds on the number of non-zero weights of two special classes of $\lambda$-constacyclic codes}
	For some special type of $\lambda$-constacyclic code $\mathcal{C}$, we can find a subgroup $\mathcal{G}$ of the automorphism group $\rm{Aut}(\mathcal{C})$ that is larger than $\langle \mu_{q}, \rho, \sigma_{\xi} \rangle$.
	According to Burnside's lemma, it is possible to obtain a smaller upper bound than $|\langle\mu_{q},\rho,\sigma_{\xi} \rangle\backslash \mathcal{C}^{*}|$ on the number of non-zero weights of $\mathcal{C}$ by counting the number of $\mathcal{G}$-orbits of $\mathcal{C}^{*}$. In this subsection, two classes of such $\lambda$-constacyclic codes are presented.

	\subsubsection{New upper bound on the number of non-zero weights of the negacyclic code $\mathcal{C}=\mathcal{R}^{(q)}_{n,\lambda}\varepsilon_{t}\bigoplus \mu_{-1}(\mathcal{R}^{(q)}_{n,\lambda}\varepsilon_{t})$}
	Assume that $q$ is a power of an odd prime and $\lambda = -1$, that is, $r=2$.
	For $0\leq t\leq s$, suppose that the irreducible $[n,k]$ negacyclic code $\mathcal{R}^{(q)}_{n,\lambda}\varepsilon_{t}$ corresponds to the $q$-cyclotomic coset $\Gamma_{t}=\{1+ra_{t},(1+ra_{t})q,\cdots,(1+ra_{t})q^{k-1}\}$.
	Since $-1\in \mathbb{Z}^{*}_{rn} \cap (1+r\mathbb{Z}_{rn})$, $\mu_{-1}$ is an $\mathbb{F}_{q}$-vector space automorphism of $\mathcal{R}^{(q)}_{n,\lambda}$.
	One can check that $\mu_{-1}(\mathcal{R}^{(q)}_{n,\lambda}\varepsilon_{t})$ is also an irreducible negacyclic code, and the primitive idempotent generating $\mu_{-1}(\mathcal{R}^{(q)}_{n,\lambda}\varepsilon_{t})$ corresponds to the $q$-cyclotomic coset $-\Gamma_{t}=\{-(1+ra_t),-(1+ra_t)q,\cdots,-(1+ra_t)q^{k-1}\}$
	(\cite{24}).
	Therefore,
	$$\mu_{-1}(\mathcal{R}^{(q)}_{n,\lambda}\varepsilon_{t})=\Big\{\sum_{j=0}^{k-1}\big( \sum_{v=0}^{k-1}c_{v}'\zeta^{-{v}(1+ra_t)q^{j}}\big)e_{-(1+ra_t)q^{j}}~\Big|~c_{v}'\in \mathbb{F}_{q}, 0\leq v \leq k-1\Big\}.$$
	
	Suppose $-(1+ra_t)\notin \Gamma_{t}$, then
	$\mu_{-1}(\mathcal{R}^{(q)}_{n,\lambda}\varepsilon_{t})\cap \mathcal{R}^{(q)}_{n,\lambda}\varepsilon_{t}=\{{\bf 0}\}$ and $\mu_{-1}^{2}(\mathcal{R}^{(q)}_{n,\lambda}\varepsilon_{t})=\mathcal{R}^{(q)}_{n,\lambda}\varepsilon_{t}$.
	Let	
	$$\mathcal{C}=\mathcal{R}^{(q)}_{n,\lambda}\varepsilon_{t}\bigoplus \mu_{-1}(\mathcal{R}^{(q)}_{n,\lambda}\varepsilon_{t}).$$
	It is easy to see that $\mu_{-1}\in {\rm Aut}(\mathcal{C})$ and the subgroup $\langle \mu_{-1}\rangle$ of ${\rm Aut}(\mathcal{C})$ is of order $2$.
	Obviously, $\langle\mu_{q},\rho,\sigma_{\xi} \rangle$ is a subgroup of $\langle\mu_{-1}, \mu_{q},\rho,\sigma_{\xi} \rangle$, so the number of $\langle\mu_{-1},\mu_{q},\rho,\sigma_{\xi} \rangle$-orbits of $\mathcal{C}^{*}=\mathcal{C}\backslash\{0\}$ is less than or equal to the number of $\langle\mu_{q},\rho,\sigma_{\xi} \rangle$-orbits of $\mathcal{C}^{*}$.
	In the following, the number of $\langle\mu_{-1},\mu_{q},\rho,\sigma_{\xi} \rangle$-orbits of $\mathcal{C}^{*}$ is given.
	
	\begin{lem}\label{3.16}
		With the notation given above. The subgroup $\langle\mu_{-1},\mu_{q},\rho,\sigma_{\xi} \rangle$ of ${\rm Aut}(\mathcal{C})$ is of order $2mn(q-1)$, and each element of $\big\langle \mu_{-1},\mu_{q},\rho,\sigma_{\xi} \big\rangle$ can be written uniquely as a product $\mu_{-1}^{r_{0}}\mu_{q}^{r_{1}}\rho^{r_{2}}\sigma_{\xi}^{r_{3}}$, where $0 \leq r_{0} \leq 1$, $0 \leq r_{1} \leq m-1$, $0 \leq r_{2} \leq n-1$ and $0 \leq r_{3} \leq q-2$.
	\end{lem}
	\begin{proof}
		The proof is similar to that of Lemma \ref{l2.2} and thus omitted here.
	\end{proof}
	Note that the group $\big\langle \mu_{-1},\mu_{q},\rho,\sigma_{\xi} \big\rangle$ can act on the sets $\mathcal{C}':=\big(\mathcal{R}^{(q)}_{n,\lambda}\varepsilon_{t}\backslash \{{\bf 0}\}\big)\cup \big(\mu_{-1}(\mathcal{R}^{(q)}_{n,\lambda}\varepsilon_{t})\backslash \{{\bf 0}\}\big)$ and $\mathcal{C}^{\sharp}:=\mathcal{R}^{(q)}_{n,\lambda}\varepsilon_{t}\backslash \{{\bf 0}\}\bigoplus \mu_{-1}(\mathcal{R}^{(q)}_{n,\lambda}\varepsilon_{t})\backslash \{{\bf 0}\}$, respectively.
	The numbers of $\big\langle \mu_{-1},\mu_{q},\rho,\sigma_{\xi} \big\rangle$-orbits of $\mathcal{C}'$ and $\mathcal{C}^{\sharp}$ are given as below.
	\begin{lem}\label{x3.2}
		With the notation given above, then the number of $\big\langle \mu_{-1},\mu_{q},\rho,\sigma_{\xi} \big\rangle$-orbits of $\mathcal{C}'$ is equal to
		\begin{align*}
			\frac{1}{k}\sum_{h|k}\varphi\big(\frac{k}{h}\big){\rm gcd}\big(q^{h}-1,\frac{q^{k}-1}{q-1},\frac{(1+ra_t)(q^{k}-1)}{rn}\big).
		\end{align*}
	\end{lem}
	\begin{proof}
		Note that
		$\mathcal{C}'=\big(\mathcal{R}^{(q)}_{n,\lambda}\varepsilon_{t}\backslash \{{\bf 0}\}\big)\cup \big(\mu_{-1}(\mathcal{R}^{(q)}_{n,\lambda}\varepsilon_{t})\backslash \{{\bf 0}\}\big)$
		is a disjoint union. Then it follows from Equation (\ref{e2.1}) and Lemma \ref{3.16} that
		\begin{align*}
			\big|\langle \mu_{-1},\mu_{q},\rho,\sigma_{\xi} \rangle\backslash \mathcal{C}'\big|
			&=\frac{1}{2mn(q-1)}\Big( \sum_{r_{0}=0}^{1}\sum_{r_{1}=0}^{m-1}\sum_{r_{2}=0}^{n-1}\sum_{r_{3}=0}^{q-2}\Big|\big\{{\bf c}\in \mathcal{R}^{(q)}_{n,\lambda}\varepsilon_{t}\backslash \{{\bf 0}\}~\big|~\mu_{-1}^{r_{0}}\mu_{q}^{r_{1}}\rho^{r_{2}}\sigma_{\xi}^{r_{3}}({\bf c})={\bf c}\big\}\Big|\\
			&\quad +\sum_{r_{0}=0}^{1}\sum_{r_{1}=0}^{m-1}\sum_{r_{2}=0}^{n-1}\sum_{r_{3}=0}^{q-2}\Big|\big\{{\bf c}\in \mu_{-1}(\mathcal{R}^{(q)}_{n,\lambda}\varepsilon_{t})\backslash \{{\bf 0}\}~\big|~\mu_{-1}^{r_{0}}\mu_{q}^{r_{1}}\rho^{r_{2}}\sigma_{\xi}^{r_{3}}({\bf c})={\bf c}\big\}\Big|\Big) .
		\end{align*}
		
		Let ${\bf c}\in \mathcal{R}^{(q)}_{n,\lambda}\varepsilon_{t}\backslash \{\bf 0\}$.
		If $r_{0}=1$, then $\mu_{-1}^{r_{0}}\mu_{q}^{r_{1}}\rho^{r_{2}}\sigma_{\xi}^{r_{3}}({\bf c})=\mu_{-1}\mu_{q}^{r_{1}}\rho^{r_{2}}\sigma_{\xi}^{r_{3}}({\bf c})\in \mu_{-1}(\mathcal{R}^{(q)}_{n,\lambda}\varepsilon_{t})\backslash \{{\bf 0}\}$, and so $\mu_{-1}^{r_{0}}\mu_{q}^{r_{1}}\rho^{r_{2}}\sigma_{\xi}^{r_{3}}({\bf c})\neq {\bf c}$.
		If $r_{0}=0$, then $\mu_{-1}^{r_{0}}\mu_{q}^{r_{1}}\rho^{r_{2}}\sigma_{\xi}^{r_{3}}({\bf c})=\mu_{q}^{r_{1}}\rho^{r_{2}}\sigma_{\xi}^{r_{3}}({\bf c})\in \mathcal{R}^{(q)}_{n,\lambda}\varepsilon_{t}\backslash \{\bf 0\}$.
		Combining these facts and the proof of Theorem \ref{t3.1}, we have
		\begin{align*}
			&\sum_{r_{0}=0}^{1}\sum_{r_{1}=0}^{m-1}\sum_{r_{2}=0}^{n-1}\sum_{r_{3}=0}^{q-2}
			\Big|\big\{{\bf c}\in \mathcal{R}^{(q)}_{n,\lambda}\varepsilon_{t}\backslash \{{\bf 0}\}~\big|
			~\mu_{-1}^{r_{0}}\mu_{q}^{r_{1}}\rho^{r_{2}}\sigma_{\xi}^{r_{3}}({\bf c})={\bf c}\big\}\Big|\\
			=&\sum_{r_{1}=0}^{m-1}\sum_{r_{2}=0}^{n-1}\sum_{r_{3}=0}^{q-2}
			\Big|\big\{{\bf c}\in \mathcal{R}^{(q)}_{n,\lambda}\varepsilon_{t}\backslash \{{\bf 0}\}~\big|
			~\mu_{q}^{r_{1}}\rho^{r_{2}}\sigma_{\xi}^{r_{3}}({\bf c})={\bf c}\big\}\Big|\\
			=&\sum_{r_{1}=0}^{m-1}{\rm gcd}\Big(n(q-1)(q^{{\rm gcd}(k,r_{1})}-1),n(q^{k}-1),\frac{(1+r{a}_t)(q-1)(q^{k}-1)}{r}\Big).
		\end{align*}
		Similar discussion as above shows that		
		\begin{align*}
			&\sum_{r_{0}=0}^{1}\sum_{r_{1}=0}^{m-1}\sum_{r_{2}=0}^{n-1}\sum_{r_{3}=0}^{q-2}\Big|\big\{{\bf c}\in \mu_{-1}(\mathcal{R}^{(q)}_{n,\lambda}\varepsilon_{t})\backslash \{{\bf 0}\}~\big|
			~\mu_{-1}^{r_{0}}\mu_{q}^{r_{1}}\rho^{r_{2}}\sigma_{\xi}^{r_{3}}({\bf c})={\bf c}\big\}\Big|\\
			=&\sum_{r_{1}=0}^{m-1}\sum_{r_{2}=0}^{n-1}\sum_{r_{3}=0}^{q-2}
			\Big|\big\{{\bf c}\in \mu_{-1}(\mathcal{R}^{(q)}_{n,\lambda}\varepsilon_{t})\backslash \{{\bf 0}\}~\big|
			~\mu_{q}^{r_{1}}\rho^{r_{2}}\sigma_{\xi}^{r_{3}}({\bf c})={\bf c}\big\}\Big|\\
			=&\sum_{r_{1}=0}^{m-1}{\rm gcd}\Big(n(q-1)(q^{{\rm gcd}(k,r_{1})}-1),n(q^{k}-1),\frac{(1+r{a}_t)(q-1)(q^{k}-1)}{r}\Big).
		\end{align*}
		Therefore, we conclude that
		\begin{align*}
			\big|\langle\mu_{-1},\mu_{q},\rho,\sigma_{\xi} \rangle\backslash \mathcal{C}'\big|
			=&\frac{2}{2mn(q-1)}\sum_{r_{1}=0}^{m-1}
			{\rm gcd}\Big(n(q-1)(q^{{\rm gcd}(k,r_{1})}-1),n(q^{k}-1),\frac{(1+r{a}_t)(q-1)(q^{k}-1)}{r}\Big)\\
			=&\frac{1}{k}\sum_{h\mid k}\varphi\big(\frac{k}{h}\big){\rm gcd}\Big(q^{h}-1,\frac{q^{k}-1}{q-1},\frac{(1+r{a}_t)(q^{k}-1)}{rn}\Big).
		\end{align*}
		Then the proof is completed.
	\end{proof}
	
	\begin{lem}\label{x3.3}
		With the notation given above, then the number of $\langle\mu_{-1},\mu_{q},\rho,\sigma_{\xi} \rangle$-orbits of $\mathcal{C}^{\sharp}$ is equal to
		\begin{align*}
			&\frac{1}{2m}\sum_{h=0}^{m-1}\left({\rm gcd}\big(q^{{\rm gcd}(k,2h)}-1,
			\frac{2(q^{k}-1)}{q-1},\frac{(1+ra_t)(q^{h}-1)(q^{k}-1)}{rn}\big)\right. \\
			&+\left. {\rm gcd}\Big((q^{{\rm gcd}(k,h)}-1)\cdot
			{\rm gcd}\big(q^{{\rm gcd}(k,h)}-1,\frac{q^{k}-1}{q-1},\frac{(1+ra_t)(q^{k}-1)}{rn}\big),
			\frac{2(1+ra_t)(q^{k}-1)^{2}}{rn(q-1)}\Big)\right).
		\end{align*}
	\end{lem}
	\begin{proof}
		According to Equation (\ref{e2.1}), Lemmas \ref{l2.2} and \ref{3.16}, we have
		$$\big|\langle\mu_{-1},\mu_{q},\rho,\sigma_{\xi} \rangle\backslash \mathcal{C}^{\sharp}\big|==\frac{1}{2mn(q-1)}\sum_{r_{0}=0}^{1}\sum_{r_{1}=0}^{m-1}\sum_{r_{2}=0}^{n-1}\sum_{r_{3}=0}^{q-2}\big| {\rm Fix}\big(\mu_{-1}^{r_{0}}\mu_{q}^{r_{1}}\rho^{r_{2}}\sigma_{\xi}^{r_{3}}\big)\big|,$$
		where $ {\rm Fix}\big(\mu_{-1}^{r_{0}}\mu_{q}^{r_{1}}\rho^{r_{2}}\sigma_{\xi}^{r_{3}}\big)=\big\{{\bf c}\in \mathcal{C}^{\sharp}~\big|~\mu_{-1}^{r_{0}}\mu_{q}^{r_{1}}\rho^{r_{2}}\sigma_{\xi}^{r_{3}}({\bf c})={\bf c}\big\}.$
		
		Take ${\bf c}={\bf c}_{t}+{\bf c}_{t}'\in \mathcal{C}^{\sharp}$, where ${\bf c}_{t}\in \mathcal{R}^{(q)}_{n,\lambda}\varepsilon_{t}\backslash \{{\bf 0}\}$ and ${\bf c}_{t}'\in \mu_{-1}(\mathcal{R}^{(q)}_{n,\lambda}\varepsilon_{t})\backslash \{{\bf 0}\}$. Let
		$${\bf c}_{t}=\sum_{j=0}^{k-1}\big(\sum_{v=0}^{k-1}c_{v}\zeta^{v (1+ra_t)q^{j}}\big)e_{(1+ra_t)q^{j}},~~{\bf c}_{t}'=\sum_{j=0}^{k-1}\big(\sum_{v=0}^{k-1}c_{v}'\zeta^{-v (1+ra_t)q^{j}}\big)e_{-(1+ra_t)q^{j}}.$$
		Suppose $r_{0}=1$.
		One can easily check that
		$\mu_{-1}(e_{(1+ra_t)q^j})=e_{-(1+ra_t)q^j}$.
		Thus
		\begin{align*}
			\mu_{-1}\mu_{q}^{r_{1}}\rho^{r_{2}}\sigma_{\xi}^{r_{3}}({\bf c})
			=&\mu_{-1}\mu_{q}^{r_{1}}\rho^{r_{2}}\sigma_{\xi}^{r_{3}}({\bf c}_{t})
			+\mu_{-1}^{r_{0}}\mu_{q}^{r_{1}}\rho^{r_{2}}\sigma_{\xi}^{r_{3}}({\bf c}_{t}')\\
			=&\sum_{j=0}^{k-1}\xi^{r_{3}}\zeta^{(1+ra_t)q^{r_{1}+j}r_{2}}\big(\sum_{v=0}^{k-1}
			c_{v}\zeta^{v (1+ra_t)q^{j}}\big)^{q^{r_{1}}}e_{-(1+ra_t)q^{j}}\\
			&+\sum_{j=0}^{k-1}\xi^{r_{3}}\zeta^{-(1+ra_t)q^{r_{1}+j}r_{2}}\big(\sum_{v=0}^{k-1}c_{v}'\zeta^{-v (1+ra_t)q^{j}}\big)^{q^{r_{1}}}e_{(1+ra_t)q^{j}}.
		\end{align*}	
		Then
		\begin{align*}
			\mu_{-1}\mu_{q}^{r_{1}}\rho^{r_{2}}\sigma_{\xi}^{r_{3}}({\bf c})={\bf c}~
			&\Leftrightarrow~
			\mu_{-1}\mu_{q}^{r_{1}}\rho^{r_{2}}\sigma_{\xi}^{r_{3}}({\bf c}_{t})={\bf c}_{t}'
			~{\rm and}~\mu_{-1}\mu_{q}^{r_{1}}\rho^{r_{2}}\sigma_{\xi}^{r_{3}}({\bf c}_{t}')={\bf c}_{t}\\
			&\Leftrightarrow~\xi^{r_{3}}\zeta^{(1+ra_t)q^{r_{1}}r_{2}}\big(\sum_{v=0}^{k-1}c_{v}\zeta^{v (1+ra_t)}\big)^{q^{r_{1}}} =\sum_{v=0}^{k-1}c_{v}'\zeta^{-v(1+ra_t)}\\
			&\qquad {\rm and}~ \xi^{r_{3}}\zeta^{-(1+ra_t)q^{r_{1}}r_{2}}\big(\sum_{v=0}^{k-1}c_{v}'\zeta^{-v (1+ra_t)}\big)^{q^{r_{1}}}=\sum_{v=0}^{k-1}c_{v}\zeta^{v(1+ra_t)}.
		\end{align*}
		Hence
		\begin{align*}
			&\big| {\rm Fix}\big(\mu_{-1}\mu_{q}^{r_{1}}\rho^{r_{2}}\sigma_{\xi}^{r_{3}}\big)\big|\\
			=&\Big|\big\{(\alpha,\beta)\in \mathbb{F}_{q^{k}}^{*} \times \mathbb{F}_{q^{k}}^{*}~\big|~ \xi^{r_{3}}\zeta^{(1+ra_t)q^{r_{1}}r_{2}}\alpha^{q^{r_{1}}}=\beta, ~\xi^{r_{3}}\zeta^{-(1+ra_t)q^{r_{1}}r_{2}}\beta^{q^{r_{1}}}=\alpha \big\}\Big|\\
			=&\Big|\big\{\alpha \in \mathbb{F}_{q^{k}}^{*}~\big|~\xi^{2r_{3}}\alpha^{q^{2r_{1}}-1}=\zeta^{-(1+ra_t)(q^{r_{1}}-1)q^{r_{1}}r_{2}}\big\}\Big|.
		\end{align*}
		It's easy to verify that $\big| {\rm Fix}\big(\mu_{-1}\mu_{q}^{r_{1}}\rho^{r_{2}}\sigma_{\xi}^{r_{3}}\big)\big|=0$ or $q^{{\rm gcd}(k,2r_{1})}-1$.
		Next, fixing $r_1~(0 \leq r_1 \leq m-1)$, we count the number of number pairs $(r_2,r_3)$ such that $\big|{\rm Fix}\big(\mu_{-1}\mu_{q}^{r_{1}}\rho^{r_{2}}\sigma_{\xi}^{r_{3}}\big)\big|\neq 0$ , where $0 \leq r_{2} \leq n-1$, $0 \leq r_{3} \leq q-2$. Let $\mathbb{F}_{q^{k}}^{*}=\langle \theta \rangle$ and denote
		$$S''(r_{1})=\big\{0\leq z\leq n-1~\big|~ \zeta^{-(1+ra_t)(q^{r_{1}}-1)q^{r_{1}}z}\in \langle \xi^2 \rangle\langle \theta^{q^{2r_{1}}-1}\rangle\big\}.$$
		Assume that $r_{2}\in S''(r_{1})$ and denote
		$$R''(r_{1},r_{2})=\big\{0\leq z\leq q-2~\big|~\zeta^{-(1+ra_t)(q^{r_{1}}-1)q^{r_{1}}r_{2}}\in \xi^{2z}\langle \theta^{q^{2r_{1}}-1}\rangle\big\}.$$
		Similar to the proof of Theorem \ref{t3.1}, we have
		$$|S''(r_{1})|={\rm gcd}\Big(n,\frac{(1+ra_t)(q^{r_{1}}-1)|\langle \xi^2 \rangle \langle \theta^{q^{2r_{1}}-1} \rangle|}{r}\Big),~
		|R''(r_{1},r_{2})|={\rm gcd}(q-1,2)\!\cdot\!|\langle \xi^2 \rangle \cap \langle \theta^{q^{2r_{1}}-1}\rangle|.$$
		It follows that
		\begin{small}
			\begin{align*}
				&\sum_{r_{1}=0}^{m-1}\sum_{r_{2}=0}^{n-1}\sum_{r_{3}=0}^{q-2}
				\Big|{\rm Fix}\big(\mu_{-1}\mu_{q}^{r_{1}}\rho^{r_{2}}\sigma_{\xi}^{r_{3}}\big)\Big|\\
				=&\sum_{r_{1}=0}^{m-1}\sum_{{\small r_{2}\in S''(r_1)}}\sum_{{\small r_{3}\in R''(r_1,r_2)}}(q^{{\rm gcd}(k,2r_{1})}-1)\\
				=&\sum_{r_{1}=0}^{m-1}|S''(r_1)|\!\cdot\!|R''(r_1,r_2)|(q^{{\rm gcd}(k,2r_{1})}-1)\\
				=&\sum_{r_{1}=0}^{m-1}{\rm gcd}(q-1,2)\cdot {\rm gcd}\Big(n|\langle \xi^2 \rangle \cap \langle \theta^{q^{2r_{1}}-1}\rangle|,\frac{(1+ra_t)(q^{r_1}-1)|\langle \xi^2 \rangle |\cdot| \langle \theta^{q^{2r_{1}}-1}\rangle|}{r}\Big)(q^{{\rm gcd}(k,2r_{1})}-1)\\
				=&\sum_{r_{1}=0}^{m-1}{\rm gcd}(q-1,2)\cdot{\rm gcd}\Big(\frac{n(q-1)}{{\rm gcd}(q-1,2)},\frac{n(q^k-1)}{q^{{\rm gcd}(k,2r_{1})}-1},\frac{(1+ra_t)(q^{r_1}-1)(q-1)(q^k-1)}{r\cdot {\rm gcd}(q-1,2)(q^{{\rm gcd}(k,2r_{1})}-1)}\Big)(q^{{\rm gcd}(k,2r_{1})}-1)\\
				=&\sum_{r_{1}=0}^{m-1}{\rm gcd}\Big(n(q-1),\frac{n(q-1)(q^k-1)}{q^{{\rm gcd}(k,2r_{1})}-1},\frac{2n(q^k-1)}{q^{{\rm gcd}(k,2r_{1})}-1},\frac{(1+ra_t)(q^{r_1}-1)(q-1)(q^k-1)}{r(q^{{\rm gcd}(k,2r_{1})}-1)}\Big)(q^{{\rm gcd}(k,2r_{1})}-1)\\
				=&\sum_{r_{1}=0}^{m-1}{\rm gcd}\Big(n(q-1)(q^{{\rm gcd}(k,2r_{1})}-1),2n(q^{k}-1),\frac{(1+ra_t)(q^{r_{1}}-1)(q-1)(q^{k}-1)}{r}\Big).
			\end{align*}
		\end{small}
		
		\noindent Suppose $r_{0}=0$, then $\mu_{-1}^{r_{0}}\mu_{q}^{r_{1}}\rho^{r_{2}}\sigma_{\xi}^{r_{3}}({\bf c})=\mu_{q}^{r_{1}}\rho^{r_{2}}\sigma_{\xi}^{r_{3}}({\bf c}_{t})+\mu_{q}^{r_{1}}\rho^{r_{2}}\sigma_{\xi}^{r_{3}}({\bf c}_{t}')$. This leads to
		$$\mu_{-1}^{r_{0}}\mu_{q}^{r_{1}}\rho^{r_{2}}\sigma_{\xi}^{r_{3}}({\bf c})={\bf c}~
		\Leftrightarrow~ \mu_{q}^{r_{1}}\rho^{r_{2}}\sigma_{\xi}^{r_{3}}({\bf c}_{t})={\bf c}_{t}~{\rm and}~\mu_{q}^{r_{1}}\rho^{r_{2}}\sigma_{\xi}^{r_{3}}({\bf c}_{t}')={\bf c}_{t}'$$
		We deduce from the proof of Corollary \ref{c3.4} that
		\begin{small}
			\begin{align*}
				&\sum_{r_{1}=0}^{m-1}\sum_{r_{2}=0}^{n-1}\sum_{r_{3}=0}^{q-2}
				\Big| {\rm Fix}\big(\mu_{q}^{r_{1}}\rho^{r_{2}}\sigma_{\xi}^{r_{3}}\big)\Big|\\
				=&\sum_{r_{1}=0}^{m-1}{\rm gcd}\Big(\!n(q-1)(q^{{\rm gcd}(k,r_{1})}-1)\!\cdot\!{\rm gcd}\big(q^{{\rm gcd}(k,r_{1})}-1,\!\frac{q^{k}-1}{q-1},\!\frac{(1+ra_t)(q^{k}-1)}{rn}\big),\!\frac{(1+ra_t)(q^{k}-1)^{2}}{r}\Big).
			\end{align*}
		\end{small}		
		
		To sum up, we have
		\begin{small}
			\begin{align*}
				\qquad &\big|\langle\mu_{-1},\mu_{q},\rho,\sigma_{\xi} \rangle\backslash \mathcal{C}^{\sharp}\big|\\
				=&\frac{1}{2mn(q-1)}\left(\sum_{r_{1}=0}^{m-1}{\rm gcd}\big(n(q-1)(q^{{\rm gcd}(k,2 r_{1})}-1),
				2n(q^{k}-1),\frac{(1+ra_t)(q^{r_{1}}-1)(q-1)(q^{k}-1)}{r}\big)\right. \\
				&+\left.\!\sum_{r_{1}=0}^{m-1}{\rm gcd}\Big(n(q-1)(q^{{\rm gcd}(k,r_{1})}-1)\!\cdot\!{\rm gcd}\big(q^{{\rm gcd}(k,r_{1})}-1,\!\frac{q^{k}-1}{q-1},\!\frac{(1+ra_t)(q^{k}-1)}{rn}\big),\!\frac{(1+ra_t)(q^{k}-1)^{2}}{r}\Big)\!\!\right)\\
				=&\frac{1}{2m}\sum_{r_{1}=0}^{m-1}\left({\rm gcd}\big(q^{{\rm gcd}(k,2r_{1})}-1,
				\frac{2(q^{k}-1)}{q-1},\frac{(1+ra_t)(q^{r_{1}}-1)(q^{k}-1)}{rn}\big)\right. \\
				&+\left. {\rm gcd}\Big((q^{{\rm gcd}(k,r_{1})}-1)\!\cdot\!
				{\rm gcd}\big(q^{{\rm gcd}(k,r_{1})}-1,\frac{q^{k}-1}{q-1},\frac{(1+ra_t)(q^{k}-1)}{rn}\big),
				\frac{(1+ra_t)(q^{k}-1)^{2}}{rn(q-1)}\Big)\right).
			\end{align*}
		\end{small}
		
		\noindent The proof is then completed.
	\end{proof}
	
	By virtue of Lemmas \ref{x3.2} and \ref{x3.3}, the number of $\langle\mu_{-1},\mu_{q},\rho,\sigma_{\xi} \rangle$-orbits of $\mathcal{C}^{*}=\mathcal{C}\backslash \{\bf 0\}$ can be immediately obtained.
	
	\begin{Theorem}\label{t3.4}
		Suppose that $\lambda=-1$ {\rm (}i.e., $r=2${\rm )}, the primitive idempotent $\varepsilon_{t}$ {\rm (}$0 \leq t \leq s${\rm )} corresponds to the $q$-cyclotomic coset $\Gamma_{t}=\{1+ra_t, (1+ra_t)q,\cdots,(1+ra_t)q^{k-1}\}$ and $-(1+ra_t)\notin \Gamma_{t}$. Let $\mathcal{C}=\mathcal{R}^{(q)}_{n,\lambda}\varepsilon_{t}\bigoplus\mu_{-1}(\mathcal{R}^{(q)}_{n,\lambda}\varepsilon_{t})$.
		Then the number of $\langle\mu_{-1},\mu_{q},\rho,\sigma_{\xi} \rangle$-orbits of $\mathcal{C}^{*}=\mathcal{C}\backslash \{\bf 0\}$ is equal to
		\begin{small}
			\begin{align*}
				&\frac{1}{k}\sum_{h|k}\varphi\big(\frac{k}{h}\big){\rm gcd}\big(q^{h}-1,\frac{q^{k}-1}{q-1},\frac{(1+ra_t)(q^{k}-1)}{rn}\big)\\
				&+\frac{1}{2m}\sum_{h=0}^{m-1}\left({\rm gcd}\big(q^{{\rm gcd}(k,2h)}\!-\!1,
				\frac{2(q^{k}\!-\!1)}{q\!-\!1},\frac{(1\!+\!ra_t)(q^{h}\!-\!1)(q^{k}\!-\!1)}{rn}\big)\right.\\
				&+\left. {\rm gcd}\Big((q^{{\rm gcd}(k,h)}-1)\!\cdot\!
				{\rm gcd}\big(q^{{\rm gcd}(k,h)}-1,\frac{q^{k}-1}{q-1},\frac{(1+ra_t)(q^{k}-1)}{rn}\big),
				\frac{(1+ra_t)(q^{k}-1)^{2}}{rn(q-1)}\Big)\!\right).
			\end{align*}
		\end{small}
		
		\noindent In particular, the number of non-zero weights of $\mathcal{C}$ is less than or equal to the number of $\langle\mu_{-1},\mu_{q},\rho,\sigma_{\xi} \rangle$-orbits of $\mathcal{C}^{*}$, with equality if and only if for any two codewords ${\bf c}_{1},{\bf c}_{2}\in \mathcal{C}^{*}$ with the same weight, there exist integers $j_{0}$, $j_{1}$, $j_{2}$ and $j_{3}$ such that
		$\mu_{-1}^{j_{0}}\mu_{q}^{j_{1}}\rho^{j_{2}}{\sigma_{\xi}}^{j_{3}}({\bf c}_{1})={\bf c}_{2}$,
		where $0 \leq j_{0} \leq 1$, $0 \leq j_{1} \leq m-1$, $0 \leq j_{2} \leq n-1$ and $0 \leq j_{3} \leq q-2$..
	\end{Theorem}
	
	\begin{proof}
		Note that
		$\mathcal{C}\backslash \{{\bf 0}\}=\mathcal{C}'\cup \mathcal{C}^{\sharp}$
		is a disjoint union. Then
		$$\big|\langle\mu_{-1},\mu_{q},\rho,\sigma_{\xi} \rangle\backslash \mathcal{C}^{*}\big|=\big|\big\langle \mu_{-1},\mu_{q},\rho,\sigma_{\xi} \big\rangle\big\backslash \mathcal{C}'\big|+\big|\big\langle \mu_{-1},\mu_{q},\rho,\sigma_{\xi} \big\rangle\big\backslash \mathcal{C}^{\sharp}\big|.$$
		The rest of the proof is clear with the help of Lemmas \ref{x3.2} and \ref{x3.3}.
	\end{proof}
	
	\begin{Remark}{\rm
			Let $\mathcal{C}$ be the negacyclic code in Theorem \ref{t3.4}. It follows from Corollary \ref{c3.4} and Theorem \ref{t3.4} that
			\begin{align*}
				\big|\langle\mu_{q},\rho,\sigma_{\xi} \rangle\backslash \mathcal{C}^{*}\big|-\big|\langle\mu_{-1},\mu_{q},\rho,\sigma_{\xi} \rangle\backslash \mathcal{C}^{*}\big|\geq & \big|\langle\mu_{q},\rho,\sigma_{\xi} \rangle\backslash \mathcal{C}'\big|-\big|\langle\mu_{-1},\mu_{q},\rho,\sigma_{\xi} \rangle\backslash \mathcal{C}'\big|\\
				=&\frac{1}{k}\sum_{h|k}\varphi\big(\frac{k}{h}\big){\rm gcd}\big(q^{h}-1,\frac{q^{k}-1}{q-1},\frac{(1+ra_t)(q^{k}-1)}{rn}\big)).
			\end{align*}
			Hence the upper bound on the number of non-zero weights of $\mathcal{C}$ given by Theorem \ref{t3.4} is less than that given by Corollary \ref{c3.4}.}
	\end{Remark}

	\begin{Example}{\rm
			Take $q=3$, $n=40$ and $\lambda=-1$.
			Let $\ell$ be the number of non-zero weights of the negacyclic code $\mathcal{C}\!=\!
			\mathcal{R}_{n}\varepsilon_{n,\lambda}^{(q)}\!\bigoplus\!\mu_{-1}(\mathcal{R}_{n,\lambda}^{(q)}\varepsilon_{3})\!=\!\mathcal{R}_{n,\lambda}^{(q)}\varepsilon_{3}\!\bigoplus\!\mathcal{R}_{n,\lambda}^{(q)}\varepsilon_{6}$,
			where the primitive idempotents $\varepsilon_{3}$ and $\varepsilon_{6}$ correspond to the $3$-cyclotomic cosets $\Gamma_{3}=\{11, 19, 33, 57\}$ and $\Gamma_{6}=\{23, 47, 61, 69\}$, respectively.
			By Corollary \ref{c3.4}, we have $\ell \leq 25$; by Theorem \ref{t3.4}, we have $\ell \leq 19$.
			After using Magma \cite{4}, we know that the weight enumerator of $\mathcal{C}$ is $1+160x^{21}+560x^{22}+320x^{23}+640x^{24}+640x^{25}+1120x^{26}\!+\!480x^{27}+640x^{28}+400x^{29}+960x^{30}+320x^{31}+320x^{32}$, which implies that the exact value of $\ell$ is 12.}
	\end{Example}
	
	\subsubsection{New upper bound on the number of non-zero weights of the constacyclic code $\mathcal{C}_{l_{0}}=\mathcal{R}^{(q)}_{n,\lambda}\varepsilon_{t}\bigoplus \mu_{(-1)^{l_{0}}p^{\frac{e}{2}}}(\mathcal{R}^{(q)}_{n,\lambda}\varepsilon_{t})$}	
	Let $q=p^{e}$, where $p$ is a prime and $e$ is even.
	For $0\leq t\leq s$, let $\mathcal{R}^{(q)}_{n,\lambda}\varepsilon_{t}$ be an irreducible $\lambda$-constacyclic code of length $n$ over $\mathbb{F}_{q}$ whose generating idempotent corresponds to the $q$-cyclotomic coset $\Gamma_t\!=\!\{1+ra_t,(1+ra_t)q,\cdots,(1+ra_t)q^{k-1}\}$ and let $l_{0} \!\in\!\{0, 1\}$.
	In this subsection, we assume that $(-1)^{l_{0}}p^{\frac{e}{2}} \!\in\! \mathbb{Z}^{*}_{rn} \cap (1+r\mathbb{Z}_{rn})$,
	then clearly $r\,|\,p^{\frac{e}{2}}-1$ when $l_{0}\!=\!0$, and $r\,|\,p^{\frac{e}{2}}+1$ when $l_{0}\!=\!1$.
	One can check that $\mu_{(-1)^{l_{0}}p^{\frac{e}{2}}}(\mathcal{R}^{(q)}_{n,\lambda}\varepsilon_{t})$ is also an irreducible $\lambda$-constacyclic code whose generating idempotent corresponds to the $q$-cyclotomic coset $(-1)^{l_{0}}p^{-\frac{e}{2}}\{1+ra_{t},(1+ra_{t})q,\cdots,(1+ra_{t})q^{k-1}\}$ (see \cite{24}).
	Therefore, $$\mu_{(-1)^{l_{0}}p^{\frac{e}{2}}}(\mathcal{R}^{(q)}_{n,\lambda}\varepsilon_{t})=\Big\{\sum_{j=0}^{k-1}\big(\sum_{v=0}^{k-1}c_{v}'\zeta^{v(-1)^{l_{0}}p^{-\frac{e}{2}}(1+ra_{t})q^{j}}\big)e_{(-1)^{l_{0}}p^{-\frac{e}{2}}(1+ra_{t})q^{j}}~\Big|~c_{v}'\in \mathbb{F}_{q}, 0 \leq v \leq k-1\Big\}.$$
	
	Suppose $(-1)^{l_{0}}p^{-\frac{e}{2}}(1+ra_{t})\!\notin\! \Gamma_{t}$, \!then \!$\mu_{(-1)^{l_{0}}p^{\frac{e}{2}}}(\mathcal{R}^{(q)}_{n,\lambda}\varepsilon_{t})\cap \mathcal{R}^{(q)}_{n,\lambda}\varepsilon_{t}\!=\!\{{\bf 0}\}$ \!and \!$\mu_{(-1)^{l_{0}}p^{\frac{e}{2}}}^{2}(\mathcal{R}^{(q)}_{n,\lambda}\varepsilon_{t})
	\!=\!\mu_{q}(\mathcal{R}^{(q)}_{n,\lambda}\varepsilon_{t})$.
	Let
	$$\mathcal{C}_{l_{0}}=\mathcal{R}^{(q)}_{n,\lambda}\varepsilon_{t}\bigoplus \mu_{(-1)^{l_{0}}p^{\frac{e}{2}}}(\mathcal{R}^{(q)}_{n,\lambda}\varepsilon_{t}).$$
	It is easy to see that
	$\mu_{(-1)^{l_{0}}p^{\frac{e}{2}}}\!\in\!{\rm Aut}(\mathcal{C}_{l_{0}})$.
	\!Since $\mu_{q}\!=\!\mu_{(-1)^{l_{0}}p^{\frac{e}{2}}}^{2}$, $\langle\mu_{q},\!\rho,\!\sigma_{\xi} \rangle$ is a subgroup of $\langle\mu_{(-1)^{l_{0}}p^{\frac{e}{2}}},\!\rho,\!\sigma_{\xi} \rangle$.
	Hence the number of $\langle\mu_{(-1)^{l_{0}}p^{\frac{e}{2}}}, \rho, \sigma_{\xi} \rangle$-orbits of $\mathcal{C}_{l_{0}}^{*}=\mathcal{C}_{l_{0}}\backslash\{0\}$ is less than or equal to the number of $\langle\mu_{q}, \rho, \sigma_{\xi} \rangle$-orbits of $\mathcal{C}_{l_{0}}^{*}$.
	Further we will show that the former is strictly less than the latter.
	
	As a preparation, we first prove the following lemma.
	\begin{lem}\label{L3}
		Let	$m$ be the order of $q$ in $\mathbb{Z}_{rn}^{*}$ and let $m_{_{l_{0}}}$ be the order of $(-1)^{l_{0}}p^{\frac{e}{2}}$ in $\mathbb{Z}_{rn}^{*}$.
		Suppose that $(-1)^{l_{0}}p^{-\frac{e}{2}}(1+ra_{t})\notin \Gamma_{t}$, then $m_{_{l_{0}}}=2m$.
	\end{lem}
	\begin{proof}
		Since $(-1)^{l_{0}}p^{-\frac{e}{2}}(1+ra_{t})\notin \{1+ra_{t}, (1+ra_{t})q,\cdots,(1+ra_{t})q^{k-1}\}$, we then have $(-1)^{l_{0}}p^{-\frac{e}{2}}q^{l}\not\equiv 1 \pmod{rn}$, or equivalently, $\big((-1)^{l_{0}}p^{\frac{e}{2}}\big)^{2l-1}\not\equiv 1 \pmod{rn}$ for any nonnegative integer $l$, and hence $m_{_{l_{0}}}$ is even. It is easy to see that $q^{\frac{m_{_{l_{0}}}}{2}}\equiv \big((-1)^{l_{0}}p^{\frac{e}{2}}\big)^{m_{_{l_{0}}}}\equiv 1 \pmod{rn}$ and $\big((-1)^{l_{0}}p^{\frac{e}{2}}\big)^{2m}\equiv q^{m}\equiv 1 \pmod{rn}$, and so $m \,|\, \frac{m_{_{l_{0}}}}{2}$ and $m_{_{l_{0}}}\,|\, 2m$, which implies that $m_{_{l_{0}}}=2m$.
	\end{proof}
	
	\begin{Theorem}\label{t3.6}
		Suppose that the primitive idempotent $\varepsilon_{t}~(0\leq t \leq s)$ corresponds to the $q$-cyclotomic coset $\Gamma_{t}=\{1+ra_t,(1+ra_t)q,\cdots,(1+ra_t)q^{k-1}\}$
		and $(-1)^{l_{0}}p^{-\frac{e}{2}}(1+ra_{t})\notin \Gamma_{t}$, where $l_0 \in \{0,1\}$.
		Let $$\mathcal{C}_{l_{0}}=\mathcal{R}^{(q)}_{n,\lambda}\varepsilon_{t}
		\bigoplus\mu_{(-1)^{l_{0}}p^{\frac{e}{2}}}(\mathcal{R}^{(q)}_{n,\lambda}\varepsilon_{t}).$$
		Then the number of $\langle\mu_{(-1)^{l_{0}}p^{\frac{e}{2}}},\rho,\sigma_{\xi} \rangle$-orbits of $\mathcal{C}_{l_{0}}^{*}=\mathcal{C}_{l_{0}}\backslash \{\bf 0\}$ is equal to
		\begin{small}
			\begin{align*}
				&\frac{1}{k}\sum_{h|k}\varphi(\frac{k}{h}){\rm gcd}\big(q^{h}-1,\frac{q^{k}-1}{q-1},\frac{(1+ra_t)(q^{k}-1)}{rn}\big) \\
				&+\frac{1}{2m}\sum_{h=0}^{m-1}\left({\rm gcd}\big(q^{{\rm gcd}(k,2h+1)}-1,\frac{2(q^{k}-1)}{q-1},\frac{[(-1)^{l_{0}}p^{-\frac{e}{2}}+q^{h}](1+ra_{t})(q^{k}-1)}{rn}\big)\right.\\
				&+\left. {\rm gcd}\Big((q^{{\rm gcd}(k,h)}-1){\rm gcd}\big(q^{{\rm gcd}(k,h)}-1,\frac{q^{k}-1}{q-1},\frac{(1+ra_{t})(q^{k}-1)}{rn}\big), \frac{[(-1)^{l_{0}}p^{-\frac{e}{2}}-1](1+ra_{t})(q^{k}-1)^{2}}{rn(q-1)}\Big)\right).
			\end{align*}
		\end{small}
		
		\noindent In particular, the number of non-zero weights of $\mathcal{C}_{l_{0}}$ is less than or equal to the number of $\langle\mu_{(-1)^{l_{0}}p^{\frac{e}{2}}},\!
		\rho,\!\sigma_{\xi} \rangle$-orbits of $\mathcal{C}_{l_{0}}^{*}$, with equality if and only if for any two codewords ${\bf c}_{1},{\bf c}_{2}\in \mathcal{C}_{l_{0}}^{*}$ with the same weight, there exist integers $j_{1}$, $j_{2}$ and $j_{3}$ such that $\mu_{(-1)^{l_{0}}p^{\frac{e}{2}}}^{j_{1}}\rho^{j_{2}}(\xi^{j_{3}}{\bf c}_{1})={\bf c}_{2}$,
		where $0 \leq j_{1} \leq 2m-1$, $0 \leq j_{2} \leq n-1$ and $0 \leq j_{3} \leq q-2$,.
	\end{Theorem}
	
	\begin{proof}
		Denote
		$$\mathcal{C}_{l_{0}}'=\big(\mathcal{R}^{(q)}_{n,\lambda}\varepsilon_{t}\backslash \{{\bf 0}\}\big)\cup \big(\mu_{(-1)^{l_{0}}p^{\frac{e}{2}}}(\mathcal{R}^{(q)}_{n,\lambda}\varepsilon_{t})\backslash \{{\bf 0}\}\big)
		~{\rm and}~\mathcal{C}_{l_{0}}^{\sharp}=\mathcal{R}^{(q)}_{n,\lambda}\varepsilon_{t}\backslash \{{\bf 0}\}\bigoplus\mu_{(-1)^{l_{0}}p^{\frac{e}{2}}}(\mathcal{R}^{(q)}_{n,\lambda}\varepsilon_{t})\backslash \{{\bf 0}\}.$$
		Note that $\mathcal{C}_{l_{0}}\backslash \{{\bf 0}\}=\mathcal{C}_{l_{0}}'\cup \mathcal{C}_{l_{0}}^{\sharp}$
		is a disjoint union, and hence
		$$\big|\langle\mu_{(-1)^{l_{0}}p^{\frac{e}{2}}},\rho,\sigma_{\xi} \rangle\backslash \mathcal{C}_{l_{0}}^{*}\big|=\big|\langle\mu_{(-1)^{l_{0}}p^{\frac{e}{2}}},\rho,\sigma_{\xi} \rangle\backslash \mathcal{C}_{l_{0}}'\big|+\big|\langle\mu_{(-1)^{l_{0}}p^{\frac{e}{2}}},\rho,\sigma_{\xi} \rangle\backslash \mathcal{C}_{l_{0}}^{\sharp}\big|.$$
		
		For $0\leq r_1\leq 2m-1$,	it is easy to see that $\mu_{(-1)^{l_{0}}p^{\frac{e}{2}}}^{r_{1}}=\mu_{(-1)^{l_{0}}p^{\frac{e}{2}}}\mu_{q}^{\frac{r_{1}-1}{2}}$ if $r_1$ is odd and $\mu_{(-1)^{l_{0}}p^{\frac{e}{2}}}^{r_{1}}=\mu_{q}^{\frac{r_{1}}{2}}$ if $r_1$ is even; in addition, $\mu_{(-1)^{l_{0}}p^{\frac{e}{2}}}(e_{(1+ra_t)q^j})=e_{(-1)^{l_{0}}p^{-\frac{e}{2}}(1+ra_t)q^j}$. Similar calculations as in Lemmas \ref{x3.2} and \ref{x3.3} show that
		\begin{align*}
			&\big|\langle\mu_{(-1)^{l_{0}}p^{\frac{e}{2}}},\rho,\sigma_{\xi} \rangle\backslash \mathcal{C}_{l_0}'\big|\\
			=&\frac{1}{2mn(q-1)}\Big(\sum_{r_{1}=0}^{2m-1}\sum_{r_{2}=0}^{n-1}\sum_{r_{3}=0}^{q-2}\Big|\big\{{\bf c}\in \mathcal{R}_{n}\varepsilon_{t}\backslash \{{\bf 0}\}~\big|~\mu_{(-1)^{l_{0}}p^{\frac{e}{2}}}^{r_{1}}\rho^{r_{2}}\sigma_{\xi}^{r_{3}}({\bf c})={\bf c}\big\}\Big|\\
			& +\sum_{r_{1}=0}^{2m-1}\sum_{r_{2}=0}^{n-1}\sum_{r_{3}=0}^{q-2}\Big|\big\{{\bf c}\in \mu_{(-1)^{l_{0}}p^{\frac{e}{2}}}(\mathcal{R}_{n}\varepsilon_{t})\backslash \{{\bf 0}\}~\big|~\mu_{(-1)^{l_{0}}p^{\frac{e}{2}}}\rho^{r_{2}}\sigma_{\xi}^{r_{3}}({\bf c})={\bf c}\big\}\Big|\Big)\\
			=&\frac{2}{2mn(q-1)}\sum_{r_{1}=0}^{m-1}{\rm gcd}\big(n(q-1)(q^{{\rm gcd}(k,r_{1})}-1),
			n(q^{k}-1),\frac{(1+ra_t)(q-1)(q^{k}-1)}{r}\big)\\
			=&\frac{1}{m}\sum_{r_{1}=0}^{m-1}{\rm gcd}\big(q^{{\rm gcd}(k,r_{1})}-1,\frac{q^{k}-1}{q-1},\frac{(1+ra_t)(q^{k}-1)}{rn}\big)\\
			=&\frac{1}{k}\sum_{h|k}\varphi(\frac{k}{h}){\rm gcd}\big(q^{h}-1,\frac{q^{k}-1}{q-1},\frac{(1+ra_t)(q^{k}-1)}{rn}\big),
		\end{align*}
		and
		\begin{small}
			\begin{align*}
				&\big|\langle\mu_{(-1)^{l_{0}}p^{\frac{e}{2}}},\rho,\sigma_{\xi} \rangle\backslash \mathcal{C}_{l_{0}}^{\sharp}\big|\\
				=&\frac{1}{2mn(q-1)}\sum_{r_{1}=0}^{2m-1}\sum_{r_{2}=0}^{n-1}\sum_{r_{3}=0}^{q-2}\Big|\big\{{\bf c}\in \mathcal{C}_{l_{0}}^{\sharp}~\big|~\mu_{(-1)^{l_{0}}p^{\frac{e}{2}}}^{r_{1}}\rho^{r_{2}}\sigma_{\xi}^{r_{3}}({\bf c})={\bf c}\big\}\Big|\\
				=&\frac{1}{2m}\sum_{r_{1}=0}^{m-1}\!\left(\! {\rm gcd}\big(q^{{\rm gcd}(k,2r_{1}+1)}-1,\frac{2(q^{k}-1)}{q-1},\frac{[(-1)^{l_{0}}p^{-\frac{e}{2}}+q^{r_{1}}](1+ra_t)(q^{k}-1)}{rn}\big)\right.\\
				&+\!\left. {\rm gcd}\Big(\!(q^{{\rm gcd}(k,r_{1})}-1)\!\cdot\!{\rm gcd}
				\big(q^{{\rm gcd}(k,r_{1})}-1,\frac{q^{k}-1}{q-1},\frac{(1+ra_t)(q^{k}-1)}{rn}\!\big), \frac{[(-1)^{l_{0}}p^{-\frac{e}{2}}-1](1+ra_t)(q^{k}-1)^{2}}{rn(q-1)}\Big)\!\!\right).
			\end{align*}
		\end{small}
		
		\noindent The desired result then follows immediately.
	\end{proof}

	\begin{Remark}{\rm
			Let $\mathcal{C}_{l_{0}}$ be the $\lambda$-constacyclic code in Theorem \ref{t3.6}. It follows from Corollary \ref{c3.4} and Theorem \ref{t3.6} that
			\begin{align*}
				\big|\langle\mu_{q},\rho,\sigma_{\xi} \rangle\backslash \mathcal{C}_{l_{0}}^{*}\big|-\big|\langle\mu_{(-1)^{l_{0}}p^{\frac{e}{2}}},\rho,\sigma_{\xi} \rangle\backslash \mathcal{C}_{l_{0}}^{*}\big|
				\geq &\big|\langle\mu_{q},\rho,\sigma_{\xi} \rangle\backslash \mathcal{C}_{l_{0}}^{'}\big|-\big|\langle\mu_{(-1)^{l_{0}}p^{\frac{e}{2}}},\rho,\sigma_{\xi} \rangle\backslash \mathcal{C}_{l_{0}}^{'}\big|\\
				=&\frac{1}{k}\sum_{h|k}\varphi\big(\frac{k}{h}\big){\rm gcd}\big(q^{h}-1,\frac{q^{k}-1}{q-1},\frac{(1+ra_{t})(q^{k}-1)}{rn}\big).
			\end{align*}
			
			Therefore, the upper bound on the number of non-zero weights of $\mathcal{C}_{l_{0}}$ given by Theorem \ref{t3.6} is less than that given by Corollary \ref{c3.4}.}
	\end{Remark}
	
	\begin{Example}{\rm
			Take $q=9$, $n=40$ and $\lambda=2$.
			We first consider the the $\lambda$-constacyclic code $\mathcal{C}_{0}=\mathcal{R}_{n,\lambda}^{(q)}\varepsilon_{0}\bigoplus \mu_{3}(\mathcal{R}_{n,\lambda}^{(q)}\varepsilon_{0})=\mathcal{R}_{n,\lambda}^{(q)}\varepsilon_{0}\bigoplus \mathcal{R}_{n,\lambda}^{(q)}\varepsilon_{1}$, where the primitive idempotents $\varepsilon_{0}$ and $\varepsilon_{1}$ correspond to the $9$-cyclotomic cosets $\Gamma_{0}=\{1,9\}$ and $\Gamma_{1}=\{3,27\}$, respectively.
			Let $\ell_{0}$ be the number of non-zero weights of $\mathcal{C}_{0}$. By Corollary \ref{c3.4}, we have $\ell_{0}\leq 14.$ However,  using Theorem \ref{t3.6} we have $\ell_{0}\leq 9$.	 After using Magma \cite{4}, we know that the weight enumerator of $\mathcal{C}_{0}$ is $1+320x^{27}+6240x^{36}$, which implies that $\ell_{0}=2$.
			
			In addition, consider the $\lambda$-constacyclic code $\mathcal{C}_{1}=\mathcal{R}_{n,\lambda}^{(q)}\varepsilon_{0}\bigoplus \mu_{-3}(\mathcal{R}_{n,\lambda}^{(q)}\varepsilon_{0})
			=\mathcal{R}_{n,\lambda}^{(q)}\varepsilon_{0}\bigoplus \mathcal{R}_{n,\lambda}^{(q)}\varepsilon_{17}$, where the primitive idempotents $\varepsilon_{0}$ and $\varepsilon_{17}$ correspond to the $9$-cyclotomic cosets $\Gamma_{0}=\{1,9\}$ and $\Gamma_{17}=\{53,77\}$, respectively. Let $\ell_{1}$ be the number of non-zero weights of  $\mathcal{C}_{1}$. One can verify that by Corollary \ref{c3.4} we have $\ell_{1}\leq 26$, and by Theorem \ref{t3.6} we have $\ell_{1}\leq 14$. After using Magma \cite{4}, we see that the weight enumerator of $\mathcal{C}_{1}$ is $1+160x^{28}+1280x^{32}+800x^{34}+2720x^{36}+640x^{38}+960x^{40}$, which implies that $\ell_{1}=6$. }
	\end{Example}
	
	\section{Concluding remarks and future works}
	
	In this paper,
	we improve the upper bounds in \cite{21} on the number of non-zero weights of any simple-root $\lambda$-constacyclic code by replacing $\langle\rho,\sigma_{\xi} \rangle$ with larger subgroups of the automorphism group of the code.
	Firstly, by calculating the number of $\langle\mu_{q},\rho,\sigma_{\xi} \rangle$-orbits of $\mathcal{C}\backslash\{\bf 0\}$, we present an explicit upper bound on the number of non-zero weights of $\mathcal{C}$ and propose a necessary and sufficient condition for $\mathcal{C}$ to meet the upper bound.
	Many examples in this paper show that our upper bound is tight, and in some cases, it is strictly smaller than the one presented in \cite{21} (see subsections 3.1 and 3.2).
	In addition, for the constacyclic code $\mathcal{C}$ belonging to two special types, we obtain a smaller upper bound on the number of non-zero weights of $\mathcal{C}$ by substituting $\langle\mu_{q},\rho,\sigma_{\xi} \rangle$ with a larger subgroup of ${\rm Aut}(\mathcal{C})$ (see subsection 3.3).
	
	A possible direction for future work is to find new few-weight constacyclic codes based on the main results presented in this paper. It would be valuable to investigate tight upper bounds on the number of Hamming weights of repeated-root constacyclic codes.

	\noindent{\bf Acknowledgements.}
	
	\!This work was supported by National Natural Science Foundation of China under Grant
	Nos.\,12271199, 12171191 and supported by self-determined research funds of CCNU from the colleges' basic research and operation of MOE CCNU22JC001.

\end{document}